\documentclass[journal]{IEEEtran}
\IEEEoverridecommandlockouts

\usepackage{amsmath,amsfonts}
\usepackage{algorithmic}
\usepackage{array}
\usepackage[caption=false,font=normalsize,labelfont=sf,textfont=sf]{subfig}
\usepackage{textcomp}
\usepackage{url}
\usepackage{verbatim}
\usepackage{graphicx}
\def\BibTeX{{\rm B\kern-.05em{\sc i\kern-.025em b}\kern-.08em
    T\kern-.1667em\lower.7ex\hbox{E}\kern-.125emX}}
\usepackage{balance}

\usepackage{cite}

\usepackage{orcidlink}

\usepackage{glossaries}
\glsdisablehyper    

\usepackage{amsmath, amssymb, bm} 
\usepackage{amsthm}
\usepackage{dsfont}
\usepackage{dashbox}
\usepackage{cancel}

\usepackage{tikz}
\usetikzlibrary{calc}
\usetikzlibrary{angles, quotes}
\usetikzlibrary{decorations.pathreplacing,decorations.markings,shapes.geometric}

\usepackage{hyperref}
\usepackage{adjustbox}

\usepackage{siunitx}

\usepackage{makecell}
\usepackage{tabularx}

\usepackage{soul}

\definecolor{myBlack}{HTML}{565656}
\definecolor{myGreen}{HTML}{5bb900}
\definecolor{myGreen2}{HTML}{005e00}
\definecolor{myBlue}{HTML}{005ae9}
\definecolor{myBlue2}{HTML}{009be8}
\definecolor{myRed}{HTML}{d7001d}
\definecolor{myOrange}{HTML}{ffa409}
\definecolor{myPurple}{HTML}{6000e2}
\definecolor{myPink}{HTML}{F03080}

\definecolor{myGreenArrow1}{HTML}{70b500}
\definecolor{myGreenArrow2}{HTML}{0d8800}

\definecolor{Pcolor}{HTML}{828282}
\definecolor{PDcolor}{HTML}{2d2d2d}
\definecolor{MMLacolor}{HTML}{005acb}
\definecolor{MMLfastcolor}{HTML}{0086ff}
\definecolor{JMLccolor}{HTML}{084b00}
\definecolor{JMLacolor}{HTML}{70b500}
\definecolor{JMLfastcolor}{HTML}{90ec48}
\definecolor{DDcentrcolor}{HTML}{ffa409}
\definecolor{DDdistrcolor}{HTML}{d7001d}
\definecolor{SDDcentrcolor}{HTML}{d454ff}
\definecolor{SDDdistrcolor}{HTML}{6000e2}

\definecolor{UCLouvainDarkBlue}{HTML}{002C62}
\definecolor{UCLouvainMediumBlue}{HTML}{4BB1E4}
\definecolor{UCLouvainLightBlue}{HTML}{8CB0CD}
\definecolor{EPLMediumBlue}{HTML}{00529C}
\definecolor{EPLLightBlue}{HTML}{4F91CD}

\theoremstyle{plain}

\newtheorem{proposition}{Proposition}[section]

\newtheorem{remark}{Remark}[section]

\newcommand{\C}{\mathds{C}}
\newcommand{\R}{\mathds{R}}
\newcommand{\E}{\mathds{E}}

\newcommand{\I}[1]{\boldsymbol{I}_{#1}}
\newcommand{\Zero}[1]{\boldsymbol{0}_{#1}}

\newcommand{\Likelihood}{\mathcal{L}}
\newcommand{\CN}{\mathcal{CN}}
\newcommand{\Compl}{\mathcal{O}}
\newcommand{\U}{\mathcal{U}}
\newcommand{\Prob}{\mathds{P}}
\newcommand{\cdf}{\mathrm{CDF}}

\newcommand{\frob}[1]{\left\lVert #1 \right\rVert_{\mathrm{F}}}
\newcommand{\norm}[1]{\left\lVert #1 \right\rVert_{\mathrm{2}}}
\newcommand{\abs}[1]{\left\lvert #1 \right\rvert}
\DeclareMathOperator*{\argmin}{argmin}
\DeclareMathOperator*{\argmax}{argmax}
\newcommand{\Exp}[1]{\E \left\{ #1 \right\}}

\DeclareMathOperator*{\vecc}{Vec}
\DeclareMathOperator*{\Herm}{H}
\DeclareMathOperator*{\Trans}{T}

\newcommand{\jc}{\mathrm{j}}    

\newcommand{\UEpos}{\boldsymbol{x}_{\mathrm{s}}}
\newcommand{\RXpos}[1]{\boldsymbol{x}_{#1}}
\newcommand{\SRXaperture}{\theta_{\mathrm{srx}}}
\newcommand{\SRXradius}{R_{\mathrm{srx}}}
\newcommand{\Sradius}{R_{\mathrm{scene}}}

\newcommand{\numRX}{N}      
\newcommand{\numF}{Q}       
\newcommand{\numP}{P}       
\newcommand{\numD}{D}       
\newcommand{\numPD}{L}      
\newcommand{\numCP}{L_{\mathrm{cp}}}      
\newcommand{\RXindex}{n}
\newcommand{\Findex}{q}     
\newcommand{\Pindex}{p}   
\newcommand{\Dindex}{d} 
\newcommand{\PDindex}{l} 
\newcommand{\BW}{B}         
\newcommand{\T}{T}          

\newcommand{\Fspacing}{\Delta_{\mathrm{f}}}       
\newcommand{\Tspacing}{\Delta_{\mathrm{t}}}       

\newcommand{\carrierF}{f_{\mathrm{s}}}              
\newcommand{\carrierWl}{\lambda_{\mathrm{s}}}       
\newcommand{\carrierWn}{k_{\mathrm{s}}}             

\newcommand{\constSet}{\mathcal{C}}                 
\newcommand{\constMap}{\nu}                         
\newcommand{\constSize}{M}                          
\newcommand{\constVar}{\sigma_{\mathrm{s}}^2}       

\newcommand{\pilotMat}{\boldsymbol{S}_{\mathrm{P}}}                 
\newcommand{\dataMat}{\boldsymbol{S}_{\mathrm{D}}}                  
\newcommand{\symbolMat}{\boldsymbol{S}}                             
\newcommand{\dataMatTest}{\widetilde{\boldsymbol{S}}_{\mathrm{D}}}  
\newcommand{\dataVecTest}{\widetilde{\boldsymbol{s}}_{\mathrm{D}}}  
\newcommand{\symbolMatTest}{\widetilde{\boldsymbol{S}}}             
\newcommand{\dataMatEst}{\widehat{\boldsymbol{S}}_{\mathrm{D}}}     
\newcommand{\dataVecEst}{\widehat{\boldsymbol{s}}_{\mathrm{D}}}     
\newcommand{\dataTensEst}{\widehat{\boldsymbol{\mathcal{S}}}_{\mathrm{D}}}

\newcommand{\pilotObs}{\boldsymbol{\mathcal{Y}}_{\mathrm{P}}}       
\newcommand{\dataObs}{\boldsymbol{\mathcal{Y}}_{\mathrm{D}}}        
\newcommand{\Obs}{\boldsymbol{\mathcal{Y}}}                         
\newcommand{\pilotEnergy}{\mathcal{P}_{\mathrm{P}}}                 
\newcommand{\rgEnergy}{\mathcal{P}}                                 
\newcommand{\pilotObsEq}{\boldsymbol{Y}_{\mathrm{P}}}               
\newcommand{\ObsEq}{\boldsymbol{Y}}                                 
\newcommand{\ObsEqDD}{\boldsymbol{Y}_{\mathrm{DD}}}                 

\newcommand{\channelMat}{\boldsymbol{H}}            
\newcommand{\channel}{\boldsymbol{\beta}}           
\newcommand{\channelcoeff}{\boldsymbol{\gamma}}     
\newcommand{\stMat}{\boldsymbol{A}}                 
\newcommand{\channelVar}{\sigma_{\gamma}^{2}}       
    
\newcommand{\channelcoeffEst}{\widehat{\boldsymbol{\gamma}}}
\newcommand{\channelcoeffEstPilots}{\widehat{\boldsymbol{\gamma}}_{\mathrm{p}}}
\newcommand{\channelcoeffEstPilotsUn}[1]{\breve{\gamma}_{#1}}
\newcommand{\channelcoeffEstPilotsU}{\breve{\boldsymbol{\gamma}}}
\newcommand{\channelEnergy}{\Gamma}

\newcommand{\channelcoeffTest}{\widetilde{\boldsymbol{\gamma}}}
\newcommand{\channelConstruct}{\widehat{\boldsymbol{H}}_{\mathrm{p}}}
\newcommand{\channelConstructEnergy}{\Gamma_{\mathrm{p}}}

\newcommand{\channelPhase}{\phi}

\newcommand{\pilotAWGN}{\boldsymbol{\mathcal{N}}_{\mathrm{P}}}       
\newcommand{\dataAWGN}{\boldsymbol{\mathcal{N}}_{\mathrm{D}}}        
\newcommand{\noiseVar}{\sigma_{\mathrm{n}}^{2}}

\newcommand{\snr}{\mathrm{SNR}}

\newcommand{\OF}{\mathcal{J}}
\newcommand{\dataMatEstcc}[1]{{\boldsymbol{W}}_{#1}}                 
\newcommand{\dataVecEstcc}[1]{{\boldsymbol{w}}_{#1}}
\newcommand{\dataMatEstccConcat}{{\boldsymbol{W}}}
\newcommand{\dataMatEstccCov}{{\boldsymbol{R}}}
\newcommand{\dataVecEstccP}{\boldsymbol{z}}
\newcommand{\dataAugVecTest}{\widetilde{\boldsymbol{t}}}
\newcommand{\OFmax}{\lambda_{\mathrm{max}}}
\newcommand{\OFnum}{U}
\newcommand{\OFdenom}{V}
\newcommand{\ARQnum}{\boldsymbol{U}}
\newcommand{\ARQdenom}{\boldsymbol{V}}
\newcommand{\MatChangeToEVD}{\boldsymbol{M}}
\newcommand{\dataAugVecTestChanged}{\widetilde{\boldsymbol{t}}^{\star}}
\newcommand{\dataAugVecTestChangedMax}{\boldsymbol{t}_{\mathrm{max}}^{\star}}
\newcommand{\ARQnumChange}{\boldsymbol{U}^{\star}}

\newcommand{\pilotTermLL}{\Likelihood_{\mathrm{P}}}
\newcommand{\dataTermLL}{\Likelihood_{\mathrm{D}}}
\newcommand{\ToA}{\tau}
\newcommand{\TDoA}{\Delta}

\newcommand{\dataObsMat}[1]{\boldsymbol{Y}_{#1}}
\newcommand{\channelConstructVec}[1]{\boldsymbol{h}_{#1}}
\newcommand{\dataObsMatCov}[1]{\boldsymbol{C}_{#1}}
\newcommand{\numA}{A}                                           

\newcommand{\SVDMatLeft}[1]{\boldsymbol{P}_{#1}}
\newcommand{\SVDMatRight}[1]{\boldsymbol{Q}_{#1}}
\newcommand{\SVDMatVals}[1]{\boldsymbol{\Sigma}_{#1}}
\newcommand{\SVindex}{s}
\newcommand{\numSV}{S}
\newcommand{\SVDValVec}[1]{\boldsymbol{\sigma}_{#1}}

\newcommand{\UEposTest}{\widetilde{\boldsymbol{x}}_{\mathrm{s}}}
\newcommand{\UEposEst}{\widehat{\boldsymbol{x}}_{\mathrm{s}}}

\newcommand{\rmse}{\mathrm{RMSE}}

\newcommand{\rangeRes}{R_{\mathrm{res}}}

\newcommand{\JMLo}{\mathrm{JML}}
\newcommand{\JMLc}{\mathrm{JML_{c}}}
\newcommand{\JMLa}{\mathrm{JML_{a}}}
\newcommand{\JMLfast}{\mathrm{JML_{fast}}}
\newcommand{\MMLo}{\mathrm{MML}}
\newcommand{\MMLa}{\mathrm{MML_{a}}}
\newcommand{\MMLfast}{\mathrm{MML_{fast}}}
\newcommand{\Pmethod}{\textsc{Pilot}}
\newcommand{\Dmethod}{\textsc{Data}}
\newcommand{\PDmethod}{\textsc{Genie}}

\newcommand{\HDDcentrmethod}{\mathrm{HDD_{centr}}}
\newcommand{\HDDdistrmethod}{\mathrm{HDD_{distr}}}

\newcommand{\Ngrid}{N_{\mathrm{g}}}
\newcommand{\Nmc}{N_{\mathrm{mc}}}

\newcommand{\MR}[1]{\textcolor{myGreen}{\textbf{[MR]:} #1}}

\newcommand{\LV}[1]{\textcolor{myRed}{\textbf{[LV]:} #1}}
\newcommand{\JL}[1]{\textcolor{myBlue}{\textbf{[JL]:} #1}}
\newacronym{isac}{ISAC}{Integrating Sensing And Communications}
\newacronym{ff}{FF}{Far-Field}
\newacronym{nf}{NF}{Near-Field}
\newacronym{5g-nr}{5G-NR}{Fifth Generation New Radio}
\newacronym{prs}{PRS}{Positioning Reference Signal}
\newacronym{6g}{6G}{Sixth Generation}
\newacronym{pcs}{PCS}{Probabilistic Constellation Shaping}
\newacronym{aoa}{AoA}{Angle of Arrival}
\newacronym{ula}{ULA}{Uniform Linear Array}
\newacronym{toa}{ToA}{Time of Arrival}
\newacronym{tdoa}{TDoA}{Time Difference of Arrival}
\newacronym{rg}{RG}{Resource Grid}
\newacronym{tfb}{TFB}{Time-Frequency Block}
\newacronym{siso}{SISO}{single-input single-output}
\newacronym{mimo}{MIMO}{multiple-input multiple-output}

\newacronym{ofdm}{OFDM}{Orthogonal Frequency-Division Multiplexing}
\newacronym{qam}{QAM}{Quandrature Amplitude Modulation}
\newacronym{bpsk}{BPSK}{Binary Phase-Shift Keying}
\newacronym{psk}{PSK}{Phase-Shift Keying}
\newacronym{gfsk}{GFSK}{Gaussian Frequency Shift Keying}
\newacronym{acfk}{ACFK}{auto-correlation function keying}

\newacronym{ue}{UE}{User Equipment}
\newacronym{tx}{TX}{transmitter}
\newacronym{rx}{RX}{receiver}
\newacronym{srx}{SRX}{sensing receiver}
\newacronym{das}{DAS}{Distributed Antenna System}
\newacronym{cpu}{CPU}{Central Processing Unit}
\newacronym{uca}{UCA}{Uniform Circular Array}

\newacronym{los}{LoS}{Line of Sight}

\newacronym{awgn}{AWGN}{Additive White Gaussian Noise}
\newacronym{snr}{SNR}{Signal-to-Noise Ratio}

\newacronym{mml}{MML}{Marginal Maximum Likelihood}
\newacronym{iml}{IML}{Integrated Maximum Likelihood}
\newacronym{ml}{ML}{Maximum Likelihood}
\newacronym{jml}{JML}{Joint Maximum Likelihood}
\newacronym{map}{MAP}{Maximum A Posteriori}
\newacronym{mmap}{MMAP}{Marginal Maximum A Posteriori}
\newacronym{np}{NP}{Nuisance Parameter}
\newacronym[longplural=Parameters of Interest]{poi}{PoI}{Parameter of Interest}
\newacronym{iid}{i.i.d.}{independent and identically distributed}
\newacronym{pdf}{p.d.f.}{probability distribution function}
\newacronym{dd}{DD}{Decision-Directed}
\newacronym{zf}{ZF}{Zero Forcing}
\newacronym{lmmse}{LMMSE}{Linear Minimum Mean Square Error}
\newacronym{rmse}{RMSE}{Root Mean Square Error}
\newacronym{mc}{MC}{Monte Carlo}
\newacronym{ser}{SER}{Symbol Error Rate}
\newacronym{af}{AF}{Ambiguity Function}
\newacronym{elmmse}{ELMMSE}{Ergodic Linear Minimum Mean Square Error}
\newacronym{nda}{NDA}{Non-Data-Aided}
\newacronym{da}{DA}{Data-Aided}
\newacronym{zzb}{ZZB}{Ziv-Zakai Bound}
\newacronym{ls}{LS}{Least Squares}
\newacronym{hdd}{HDD}{Hard Decision-Directed}
\newacronym{sdd}{SDD}{Soft Decision-Directed}
\newacronym{dop}{DoP}{Dilution of Precision}
\newacronym{crb}{CRB}{Cramér-Rao Bound}

\newacronym{evd}{EVD}{Eigenvalue Decomposition}
\newacronym{svd}{SVD}{Singular Value Decomposition}

\newacronym{cdf}{CDF}{Cumulative Distribution Function}

\begin{document}

\title{A Computationally Efficient Joint Maximum Likelihood Estimator \\ for Passive Localization in OFDM Distributed Antenna Systems \\ with Pilots and Unknown Data Payloads
    \thanks{Mathieu Reniers is a Research Fellow of the Fonds de la Recherche Scientifique - FNRS.}
}
\author{
    Mathieu Reniers\orcidlink{0009-0007-3811-6112},
    Martin Willame\orcidlink{0000-0002-7107-6198}, 
    Jérôme Louveaux\orcidlink{0000-0003-2557-1857}, 
    Luc Vandendorpe\orcidlink{0000-0003-4958-8848}, \\
    ICTEAM, UCLouvain - Louvain-La-Neuve, Belgium.
    \footnotesize{Emails: \{firstname.lastname\}@uclouvain.be}
}

\bstctlcite{MyBSTcontrol}


\maketitle


\begin{abstract}
Communication-centric Integrated Sensing and Communications (ISAC) is a promising paradigm for sixth-generation (6G) wireless systems, enabling new sensing services by leveraging the already-deployed communication infrastructure.
Communication signals typically comprise both known deterministic pilot sequences and unknown random data payloads.
For localization and sensing tasks, the prevailing approach in multistatic and distributed ISAC systems relies exclusively on pilot symbols, entirely overlooking the positioning information carried by data payloads, which constitute the majority of each transmitted frame.
Alternatively, Decision-Directed (DD) approaches treat data estimates as additional pilots, inherently limiting localization performance to that of the underlying communication system, while Non-Data-Aided (NDA) methods from the literature require prior knowledge of the data symbol distribution and incur a computational cost that grows with constellation size.
In this paper, we derive a Joint Maximum Likelihood (JML) estimator that jointly exploits pilot and data symbols for localization without requiring data decoding, in a passive scenario where a distributed sensing receiver localizes a User Equipment (UE) by exploiting its Orthogonal Frequency-Division Multiplexing (OFDM) communication signal as a signal of opportunity.
The optimal solution is derived and shown to be computationally intractable for typical 6G parameters.
Two tractable approximations are then proposed, achieving localization performance superior to DD baselines at comparable computational complexity, while remaining constellation-agnostic and yielding substantially lower computational requirements than existing NDA approaches.
Furthermore, the proposed estimators are shown to admit a geometric interpretation, providing insight into their intrinsic localization behavior.

\end{abstract}


\section{Introduction}\label{sec:introduction}

\IEEEPARstart{S}{ixth}-generation (6G) \glsunset{6g} wireless systems identify \gls{isac} as a key feature of future networks, enabling mobile infrastructure to sense their surrounding environment \cite{itu_recommendation_2023, mazahir_survey_2021}.
Embedding both functions within a single platform enables hardware reuse, improved spectral efficiency \cite{liu_joint_2020}, and reduced device cost, form factor, and energy consumption, while potentially benefiting the performance of each service \cite{feng_joint_2020, zhang_overview_2021}.
Among existing \gls{isac} approaches, the communication-centric paradigm—where communication remains the primary function and sensing is performed by reusing the communication waveform—appears particularly promising \cite{lu_sensing_2026}, as it leverages the already-deployed communication infrastructure.
Such signals typically comprise both pilot or reference sequences and data payloads.
Pilot sequences are deterministic and known at both the \gls{tx} and \gls{rx}, and exhibit strong correlation properties leading to robust sensing performance.
Data symbols, by contrast, convey information and are therefore random and unknown to the \gls{rx} a priori.
Most existing approaches in multistatic \gls{isac} and equivalent source localization systems rely exclusively on known pilot signals, e.g., using \glsdesc{prs} embedded in the \gls{5g-nr} frame structure \cite{wei_5g_2023, golzadeh_joint_2024, wei_multiple_2024}.
However, such reference signals typically account for only \SI{3}{\percent}--\SI{25}{\percent} of each transmitted frame in \gls{5g-nr} \cite{3gpp_5g_2025}, leaving the majority of communication resources unexploited for sensing.
Since data payloads are transmitted regardless of any sensing intent, leveraging them for localization incurs no additional spectral or energy overhead, making their exploitation particularly attractive from a resource efficiency standpoint.
In this work, we investigate the joint exploitation of known pilots and unknown data payloads for opportunistic localization\footnotemark, defined here as the usage of communication signals transmitted by a third-party device as signals of opportunity to localize either the source or passive targets.
\footnotetext{The term opportunistic localization is used interchangeably with passive sensing terminology in this work.}

\subsection{Related Work}

\subsubsection{Sensing Performance Under Random Data Signals}
To address the impact of data randomness on sensing performance, multiple studies have characterized sensing using random data signals rather than fixed pilot sequences, typically in monostatic configurations assuming perfect knowledge of the transmitted data at the \gls{srx} \cite{zhang_discrete-periodic_2026}.
Keskin et al. \cite{keskin_fundamental_2025} reveal a fundamental time-frequency tradeoff in monostatic \gls{ofdm} \gls{isac} systems: high-order \gls{qam} improves communication rates but degrades sensing due to elevated sidelobe levels from the non-constant modulus, whereas low-order \gls{psk} minimizes sidelobes at the cost of reduced spectral efficiency.
Liu et al. \cite{liu_cp-ofdm_2025, liu_cp-ofdm_2026} show that \gls{ofdm} with \gls{qam} achieves superior ranging performance among cyclic-prefix waveforms, motivating its adoption in \gls{6g} and in this work.
To further mitigate the sensing-communication tradeoff, peak sidelobe control \cite{zhao_auto-correlation_2026}, constellation shaping \cite{du_reshaping_2024, keshavarzchafjiri_gamma-distributed_2026}, and mismatched filtering \cite{yang_constellation-independent_2026-1} have also been proposed.

However, while all these studies provide valuable insights into the impact of data randomness on sensing, the underlying assumption of perfect data knowledge at the \gls{srx} is not achievable in multistatic or passive sensing configurations.

\subsubsection{Decision-Directed Approaches}\label{sec:SOTA_DD}
A natural way to exploit the full transmitted frame when data symbols are unknown is to demodulate the data payloads and treat the resulting decisions as additional pilots, following the \gls{dd} philosophy.
Wypich and Zielinski \cite{wypich_ofdm-based_2025} demonstrate that under low \gls{ser} conditions, estimated data symbols can effectively supplement pilots in \gls{ofdm}-based passive radar, while performance degrades to the pilot-only case under challenging channel conditions or high-order modulations.
In \cite{wypich_experimental_2026}, the same authors experimentally validate this approach using 5G waveforms.
Brunner et al. \cite{brunner_bistatic_2025} propose a complete bistatic \glsdesc{siso} framework incorporating synchronization and frame reconstruction based on channel-coded communication, subsequently extended to \glsdesc{mimo} in \cite{giroto_system_2026}.
Henninger et al. \cite{henninger_hybrid_2026} propose a resource allocation scheme placing lower-order constellation symbols as pseudo-pilots to mitigate the sensing-communication tradeoff, while Hu et al. \cite{hu_learning-based_2025} propose an end-to-end deep learning framework for joint constellation design and data demodulation in bistatic \gls{ofdm} \gls{isac}.
The \gls{dd} methodology has also been integrated into iterative schemes that exploit the mutual benefits of \textit{sensing-aided communication} and \textit{communication-aided sensing}, first for single-carrier bistatic systems \cite{zhao_joint_2024} and subsequently extended to \gls{ofdm} \cite{keskin_bridging_2025-1}.

While these methods can achieve strong localization performance under favorable \gls{snr} conditions, they may incur high computational complexity due to forward error correction decoding or iterative processing, the latter typically relying on heuristic designs and requiring sufficiently reliable data detection \cite{chen_joint_2026}.
Moreover, the localization performance of \gls{dd} approaches is inherently limited by that of the underlying communication system.

\subsubsection{Non-Data-Aided Approaches}
An alternative to \gls{dd} is the \gls{nda} philosophy, in which the unknown transmitted symbols are treated as random variables rather than being pre-estimated.
Specifically, the data symbols are modeled as \textit{\glspl{np}}—parameters that affect the observation distribution but are not of direct interest—as opposed to the \textit{\glspl{poi}}.

A natural way to eliminate the dependence on \glspl{np} is to \textit{marginalize} the likelihood over their prior distribution, an approach referred to as \gls{mml} estimation.
In the case of data symbols, this amounts to evaluating the conditional likelihood over all constellation points, rather than making hard symbol decisions as in \gls{dd}.
Monfared et al. \cite{monfared_iterative_2020} propose an iterative \gls{mml}-based algorithm using angle-of-arrival measurements for sensor localization with a Gaussian Frequency Shift Keying waveform.
Graff and Humphreys \cite{graff_ofdm-based_2026} derive the \gls{mml} estimator for \gls{ofdm}-based positioning and validate their framework on simulated low-earth orbit satellite scenarios.
In \cite{reniers_localization_2026}, we extend this methodology to a multi-antenna distributed array with unknown complex channel coefficients, with an analytical acceleration reducing the per-symbol complexity from $\Compl(\constSize)$ to $\Compl(\sqrt{\constSize})$ for $\constSize$-\gls{qam} constellations.

However, all these works \cite{monfared_iterative_2020, graff_ofdm-based_2026, reniers_localization_2026} require prior knowledge of the constellation—and of the coding scheme if applicable—which may not be accessible at the \gls{srx} in an opportunistic scenario, and may incur high complexity due to enumeration over the constellation in the localization step.

Another approach to handling \glspl{np} is \gls{jml} estimation, where \glspl{poi} and \glspl{np} are estimated jointly, entailing a high-dimensional joint search. 
A strategy to eliminate the \glspl{np} dependence is to \textit{concentrate} the likelihood \cite{stoica_concentrated_1995}, i.e., express closed-form \glspl{np} estimates as a function of the \glspl{poi} and substitute them back into the likelihood, yielding a substantially reduced complexity.
In \cite{reniers_joint_2026-1}, we applied this methodology to a source localization scenario restricted to an \gls{srx} consisting of a co-located \gls{ula} operating under \gls{ff} conditions.

The present work extends this framework to a distributed \gls{srx} with unknown complex channel coefficients, and provides a comprehensive comparison against the \gls{mml} approach of \cite{graff_ofdm-based_2026, reniers_localization_2026}, pilot-only and \gls{dd} baselines, as well as a practical performance upper bound, in terms of both localization accuracy and computational requirements.

\subsection{Contributions}
To address the aforementioned limitations of \gls{dd} and marginal approaches, we propose a \gls{jml} framework for an ``uplink''\footnotemark\ scenario in which the communication signal transmitted by a \gls{ue} is passively captured by an opportunistic \gls{srx} consisting of multiple distributed nodes.
While this scenario is considered for simplicity, all proposed methods can be directly adapted to classical multistatic passive sensing, where the \gls{srx} localizes a passive target based on signal reflections, as discussed in Section~\ref{sec:other_scenarios}.
\footnotetext{This terminology is used by analogy with the conventional communication uplink, and is used here interchangeably with source localization.}
The main contributions of this paper are summarized as follows:
\begin{itemize}
    \item A \gls{jml} framework is developed for uplink localization with a distributed \gls{srx}, treating both random data symbols and random complex channel gains as \glspl{np}. The optimal solution is derived in closed form as a generalized Rayleigh quotient, reducible to a classical \gls{evd} problem, and is shown to be computationally intractable for typical time-frequency resources, motivating the development of practical approximations.
    \item Based on a theoretical analysis of the optimal solution, two computationally feasible approximations are proposed, offering different performance-complexity tradeoffs. The resulting estimators are \textit{constellation-agnostic}: no knowledge of the constellation is required at the \gls{srx}, and regardless of the modulation scheme employed, both their localization performance and computational requirements remain unchanged—a key advantage over \gls{dd} and marginal approaches.
    \newpage
    \item Through \gls{mc} simulations, significant localization performance improvements over \gls{dd} baselines are demonstrated at comparable computational complexity, while the proposed estimators slightly outperform the \gls{mml} approach at substantially reduced complexity.
    \item A comprehensive analysis of the impact of system parameters on both localization performance and computational requirements is provided, supported by asymptotic complexity analysis, simulation results, and runtime measurements.
    \item The proposed \gls{jml} estimators are shown to operate as single-step hybrid \gls{toa}-\gls{tdoa} methods, combining range circle loci from the pilot term and hyperbolic loci from the data term. 
    This geometric interpretation provides theoretical insight into their intrinsic localization behavior.
\end{itemize}

\subsection{Structure of the Paper}

The remainder of this paper is organized as follows.
Section~\ref{sec:system_model} describes the system model, including the scenario, channel model, and signal definitions.
Section~\ref{sec:JML} derives the optimal \gls{jml} estimator and establishes its computational intractability for typical time-frequency resources, while Section~\ref{sec:JML_heuristics} introduces two tractable approximations with different performance-complexity tradeoffs, and identifies additional scenarios to which the proposed methods directly apply.
Section~\ref{sec:results} presents numerical simulation results, demonstrates the localization improvement of the proposed estimators over existing methods, investigates the impact of system parameters on localization performance, and establishes the single-step hybrid \gls{toa}-\gls{tdoa} interpretation of the proposed estimators.
Section~\ref{sec:Complexity} provides a detailed complexity analysis based on asymptotic operation counts and practical runtime measurements.
Finally, Section~\ref{sec:conclusion} concludes the paper and outlines directions for future work.

\subsection{Notations}

Scalars, vectors, matrices, and tensors are respectively denoted by $a$, $\boldsymbol{a}$, $\boldsymbol{A}$, and $\boldsymbol{\mathcal{A}}$, and when functions of parameters, written as $a(\cdot)$, $\boldsymbol{a}(\cdot)$, $\boldsymbol{A}(\cdot)$, and $\boldsymbol{\mathcal{A}}(\cdot)$.
The $i$-th element of $\boldsymbol{a}$, the $(i,j)$-th element of $\boldsymbol{A}$, and the $(i,j,k)$-th element of $\boldsymbol{\mathcal{A}}$ are indexed as $\boldsymbol{a}[i]$, $\boldsymbol{A}[i,j]$, and $\boldsymbol{\mathcal{A}}[i,j,k]$.
The symbol ``$:$'' selects all elements along a given dimension, e.g.,  $\boldsymbol{\mathcal{A}}[i,j,:]$ denotes the vector of  $\boldsymbol{\mathcal{A}}$ along the third dimension for indices $(i,j)$.
The vectorization operator $\vecc\{\boldsymbol{A}\}$ stacks the columns of $\boldsymbol{A}$ into a single column vector.
The absolute value, the vector $\ell_2$ norm, and the Frobenius norm are written $\abs{a}$, $\norm{\boldsymbol{a}}$, and $\frob{\boldsymbol{A}}$, respectively.
The transpose, complex conjugate, and Hermitian transpose are written $\boldsymbol{A}^T$, $\boldsymbol{A}^*$, and $\boldsymbol{A}^H$.
The real and complex sets are $\R$ and $\C$, and $\jc \triangleq \sqrt{-1}$.
The zero vector of size $K$ is expressed as $\Zero{K}$, the $K \times K$ identity matrix as $\I{K}$, the asymptotic complexity  as $\Compl(\cdot)$ and the cardinality of a finite set $\mathcal{S}$ as $\#\mathcal{S}$.
The likelihood function and the expectation over $\ast$ are denoted $\Likelihood(\cdot)$ and $\E_{\ast}\{\cdot\}$, respectively.
Finally, given a true quantity $\theta$, candidate and final estimates are indicated by $\widetilde{\theta}$ and $\widehat{\theta}$, respectively.
\section{System Model}\label{sec:system_model}

\subsection{Scenario}
In this paper, we investigate the problem of localizing a single-antenna \gls{ue} at position $\UEpos \in \R^{2}$, which transmits an \gls{ofdm} uplink signal to an intended communication receiver, e.g., a base station.
At the same time, this communication signal is opportunistically captured by an \gls{srx} consisting of $\numRX$ single-antenna \gls{rx} nodes located at positions $\{\RXpos{\RXindex}\}_{\RXindex=0}^{\numRX-1} \in \R^{\numRX \times 2}$, forming a \gls{das}.
A \gls{cpu} collects the received signals from all nodes and estimates the \gls{ue} position $\UEpos$.
The considered scenario is illustrated in \autoref{fig:scenario}.
Note that the proposed framework extends straightforwardly to $\R^3$ and to other scenarios, as described in Section~\ref{sec:other_scenarios} (see \autoref{fig:other_scenarios}).

\begin{figure}[ht]
    \centering
        \resizebox{0.95\linewidth}{!}{%
        \def\SRXCOLOR{UCLouvainDarkBlue}
\def\TXCOLOR{UCLouvainMediumBlue}

\newcommand{\antenna}[2][1]{
    \draw[thick, \SRXCOLOR] #2 -- ++(0,0.5*#1);
    \draw[thick, \SRXCOLOR] #2 ++(0,0.5*#1) -- ++(-0.25*#1,0.25*#1);
    \draw[thick, \SRXCOLOR] #2 ++(0,0.5*#1) -- ++(0.25*#1,0.25*#1);
}
\def\ueX{-0.5}  
\def\ueY{0} 
\def\ueradius{0.08} 
\def\antennaheight{0.75}    
\def\antennapositions{(-3,-2), (-1,-3), (1,-2.5), (3,-1), (2,0.5)}
\def\cpuY{-4.5}  
\def\cpuX{0.5}    

\begin{tikzpicture}
    \coordinate (ue) at (\ueX,\ueY);
    \fill[\TXCOLOR] (ue) circle (\ueradius);
    \node[above, \TXCOLOR] at ($(ue)+(0,0.1)$) {TX};
    \node[below, \TXCOLOR] at ($(ue)+(0,-0.15)$) {$\UEpos$};

    \draw[\TXCOLOR!50!white] (ue) circle (0.7);
    \draw[\TXCOLOR!35!white] (ue) circle (0.9);
    \draw[\TXCOLOR!20!white] (ue) circle (1.1);
    \draw[\TXCOLOR!5!white] (ue) circle (1.3);
    \foreach [count=\i from 0] \pos in \antennapositions {
        \antenna[1.0]{\pos}
        \path \pos coordinate (ant\i);  
    }
    \foreach [count=\i from 0] \pos in \antennapositions {
            \ifnum\i<2
        \node[below, \SRXCOLOR] at \pos {\small RX$_{\i}$};
        \node[below, \SRXCOLOR] at \pos (rx\i) {\small RX$_{\i}$};
    \fi
    }
    \path let \p1=(ant1), \p2=(ant2), \n1={atan2(\y2-\y1,\x2-\x1)} in
    node[rotate=\n1] at ($(ant1)!0.5!(ant2)$) {$\cdots$};
    \path let \p1=(ant2), \p2=(ant3), \n1={atan2(\y2-\y1,\x2-\x1)} in
    node[rotate=\n1] at ($(ant2)!0.5!(ant3)$) {$\cdots$};
    \node[below, \SRXCOLOR] at (ant2) (rx2) {\small RX$_{n}$};
    \node[right, \SRXCOLOR] at ($(ant2)+(0,+0.25)$) {$\RXpos{\RXindex}$};
    \node[below, \SRXCOLOR] at (ant3) (rx3) {\small RX$_{\numRX - 2}$};
    \node[below, \SRXCOLOR] at (ant4) (rx4) {\small RX$_{\numRX - 1}$};
    
    \node[draw=\SRXCOLOR, thick, rectangle, rounded corners, fill=\SRXCOLOR!10, text=\SRXCOLOR, minimum width=2cm, minimum height=0.5cm] at (\cpuX,\cpuY) (cpu) {CPU};
    \draw[\SRXCOLOR,thin] ($(ant0)+(0,-0.5)$) -- ++(0,-1.3) -| (cpu.north);
    \draw[\SRXCOLOR,thin]  ($(ant1)+(0,-0.5)$) -- ++(0,-0.3) -| (cpu.north);
    \draw[\SRXCOLOR,thin]  ($(ant2)+(0,-0.5)$) -- ++(0,-0.8) -| (cpu.north);
    \draw[\SRXCOLOR,thin]  ($(ant3)+(0,-0.5)$) -- ++(0,-2.3) -| (cpu.north);
    \draw[\SRXCOLOR,thin]  ($(ant4)+(0,-0.5)$) -- ++(0.0,-0.1) -- ++(2.0,0.0) -- ++(0,-3.7) -| (cpu.north);

    \draw[<->,sloped] ($(ant2)+(0,\antennaheight)$) -- ($(ue)+(\ueradius,-\ueradius)$) 
            node[pos=0.5, align=center,above=0.02cm, fill=white, fill opacity=0.5, text opacity=1] {$\norm{\UEpos - \RXpos{n}}$};
    
    \draw[thick, ->, gray] (-3.1,0) -- ++(0.5,0) node[right]{$x$};
    \draw[thick, ->, gray] (-3,-0.1) -- ++(0,0.5) node[above]{$y$};

    \node[\SRXCOLOR] at (-1.8,-1.5) {SRX};

\end{tikzpicture}
    }
    \caption{Illustration of the considered scenario. 
    The \gls{ue} (i.e., the \gls{tx}) transmits an uplink signal to a \gls{das} (i.e., the \gls{srx}) consisting of $\numRX$ single-antenna \gls{rx} nodes. 
    A \gls{cpu} collects the received signals and performs localization.
    }
    \label{fig:scenario}
\end{figure}
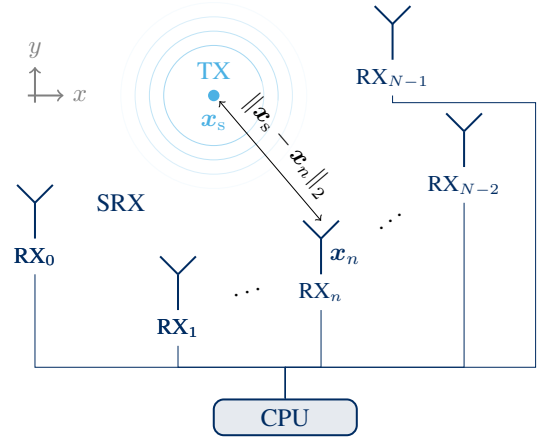

The \gls{ue} transmits \gls{ofdm} symbols comprising $\numP$ pilots and $\numD$ data symbols across $\numF$ subcarriers spaced by $\Fspacing$, yielding a total bandwidth of $\BW \triangleq \numF \Fspacing$.
The lowest subcarrier frequency is denoted $\carrierF$, with corresponding wavelength $\carrierWl$ and wavenumber $\carrierWn \triangleq 2\pi \carrierF / c = 2\pi / \carrierWl$, where $c$ is the speed of light.
Pilot symbols are denoted by $\pilotMat \in \C^{\numF \times \numP}$ and may consist of arbitrary complex sequences.
These sequences are deterministic and designed to exhibit well-defined correlation properties \cite{wypich_experimental_2026}.
Data symbols are denoted by $\dataMat \in \constSet^{\numF \times \numD} \subset \C^{\numF \times \numD}$, where each element $\dataMat[\Findex,\Dindex]$ belongs to a constellation $\constSet_{\constMap(\Findex,\Dindex)}$.
The mapping $\constMap(\Findex,\Dindex)$ assigns a constellation from the set $\constSet$ to each resource element $(\Findex, \Dindex)$. 
It is determined by the \gls{tx} for communication purposes, and is typically defined in the transmission preamble.
For notational convenience, we adopt the following shorthand: \begingroup \small$\dataMat \in \constSet_{\constMap}^{\numF \times \numD} \triangleq \{\dataMat[\Findex,\Dindex] \in \constSet_{\constMap(\Findex,\Dindex)}\}_{\Findex,\Dindex=0,0}^{\numF-1,\numD-1}$\endgroup.
Unlike pilot symbols, data symbols are random as they convey information \cite{keskin_fundamental_2025}.
The symbol transmission follows the \gls{tfb} illustrated in \autoref{fig:resource_grid}.
Note that other resource allocations can be accommodated; specifically, the time-domain distribution of pilot and data symbols has no impact, as the channel is assumed constant across the time dimension (see Section~\ref{sec:channel_model}), whereas a uniform allocation across the frequency axis is required.

\begin{figure}[ht]
    \resizebox{0.88\linewidth}{!}{%
        \def\numPilots{4}
\def\numData{14}
\pgfmathsetmacro{\numTimeSlots}{\numPilots+\numData}   
\def\numSubcarriers{8}   
\def\cellSize{0.25}       

\def\DataIndex{1}
\def\SubcarrierIndex{6}

\begin{tikzpicture}

    \pgfmathsetmacro{\numPilotsMinusOne}{\numPilots-1}
    \pgfmathsetmacro{\numTimeSlotsMinusOne}{\numTimeSlots-1}
    \pgfmathsetmacro{\numSubcarriersMinusOne}{\numSubcarriers-1}
    
    \foreach \x in {0,...,\numPilotsMinusOne} {
        \foreach \y in {0,...,\numSubcarriersMinusOne} {
            \fill[black!10] (\x*\cellSize, \y*\cellSize) rectangle ({(\x+1)*\cellSize}, {(\y+1)*\cellSize});
        }
    }
    
    \foreach \x in {\numPilots,...,\numTimeSlotsMinusOne} {
        \foreach \y in {0,...,\numSubcarriersMinusOne} {
            \fill[black!20] (\x*\cellSize, \y*\cellSize) rectangle ({(\x+1)*\cellSize}, {(\y+1)*\cellSize});
        }
    }

    \foreach \x in {0,...,\numTimeSlots} {
        \draw[black!30] (\x*\cellSize, 0) -- (\x*\cellSize, \numSubcarriers*\cellSize);
    }
    \foreach \y in {0,...,\numSubcarriers} {
        \draw[black!30] (0, \y*\cellSize) -- (\numTimeSlots*\cellSize, \y*\cellSize);
    }

    \draw[->, thick] (-0.75*\cellSize, 0) -- (\numTimeSlots*\cellSize + \cellSize, 0) node[pos=0.4, below=0.5cm] {Time};
    \draw[->, thick] (0, -0.75*\cellSize) -- (0, \numSubcarriers*\cellSize + \cellSize) node[pos=0.5, left=1.4cm, rotate=90, anchor=center] {Frequency};

    \draw[<->, thick] (0, -\cellSize) -- (\numPilots*\cellSize, -\cellSize) node[midway, below] {$\numP\Tspacing$};
    \draw[<->, thick] (\numPilots*\cellSize, -\cellSize) -- (\numTimeSlots*\cellSize, -\cellSize) node[midway, below] {$\numD\Tspacing$};
    \draw[<->, thick] (-\cellSize,0) -- (-\cellSize, \numSubcarriers*\cellSize) node[midway, left] {$\numF\Fspacing$}
    node[pos=0, left=0.1cm] {$\carrierF$};

    \pgfmathsetmacro{\pilotCenter}{\numPilots*\cellSize/2}
    \pgfmathsetmacro{\dataCenter}{(\numPilots + \numTimeSlots)*\cellSize/2}
    \pgfmathsetmacro{\verticalCenter}{\numSubcarriers*\cellSize/2}
    
    \node[] at (\pilotCenter, \verticalCenter) {$\pilotMat$};
    \node[] at (\dataCenter, \verticalCenter) {$\dataMat$};

    \draw[<->, thick]
    (\numTimeSlots*\cellSize - 2*\cellSize, \numSubcarriers*\cellSize + 0.75*\cellSize)
    --
    (\numTimeSlots*\cellSize - 1*\cellSize, \numSubcarriers*\cellSize + 0.75*\cellSize)
    node[midway, above] {$\Tspacing$};

    \draw[<->, thick]
    (\numTimeSlots*\cellSize + 0.75*\cellSize, \numSubcarriers*\cellSize - 2*\cellSize)
    --
    (\numTimeSlots*\cellSize + 0.75*\cellSize, \numSubcarriers*\cellSize - 1*\cellSize)
    node[midway, right] {$\Fspacing$};


    \def\DataIndex{1}
    \def\SubcarrierIndex{3}
    \pgfmathsetmacro{\dotX}{(\numPilots + \DataIndex + 0.5)*\cellSize}
    \pgfmathsetmacro{\dotY}{(\SubcarrierIndex + 0.5)*\cellSize}

    \fill[myBlue] (\dotX, \dotY) circle (1.5pt);

    \pgfmathsetmacro{\gridTop}{\numSubcarriers*\cellSize}
    \pgfmathsetmacro{\constellationBottom}{\gridTop + 0.3}
    \pgfmathsetmacro{\constellationCenterX}{\dotX - 2*\cellSize}
    \pgfmathsetmacro{\constellationCY}{\constellationBottom + 0.45}

    \draw[myBlue, thin, dashed] (\dotX, \dotY) node[below, font=\tiny] {$(\Findex,\Dindex)$} -- (\constellationCenterX, \constellationBottom) ;

    \def\qamScale{0.1}
    \def\qamDotSize{1.2pt}
    \pgfmathsetmacro{\qamHalf}{3*\qamScale}

    \node[above, font=\footnotesize] at (\constellationCenterX, \constellationBottom + 0.93)
        {$\constSet_{\constMap(\Findex,\Dindex)}$};

    \draw[->, black!50, thin]
        ({\constellationCenterX - \qamHalf - 0.2}, {\constellationCY})
        --
        ({\constellationCenterX + \qamHalf + 0.2}, {\constellationCY});
    \draw[->, black!50, thin]
        ({\constellationCenterX}, {\constellationBottom + 0.02})
        --
        ({\constellationCenterX}, {\constellationBottom + 0.93});

    \pgfmathsetmacro{\selectedQI}{3}    
    \pgfmathsetmacro{\selectedQQ}{3}  

    \pgfmathsetmacro{\selectedI}{2*\selectedQI - 3}
    \pgfmathsetmacro{\selectedQ}{2*\selectedQQ - 3}

    \foreach \xi in {-3,-1,1,3} {
        \foreach \yi in {-3,-1,1,3} {
            \pgfmathsetmacro{\px}{\constellationCenterX + \xi*\qamScale}
            \pgfmathsetmacro{\py}{\constellationCY + \yi*\qamScale}
            \pgfmathsetmacro{\isSelected}{(\xi == \selectedI) && (\yi == \selectedQ) ? 1 : 0}
            \ifnum\isSelected=1
                \fill[myBlue] (\px, \py) circle (2pt);
                \draw[myBlue, thin] (\px, \py) circle (3pt) node[above right=1pt, font=\footnotesize] {$\dataMat[\Findex,\Dindex]$};
            \else
                \fill[black] (\px, \py) circle (\qamDotSize);
            \fi
        }
    }


    \def\DataIndex{9}
    \def\SubcarrierIndex{6}
    \pgfmathsetmacro{\dotX}{(\numPilots + \DataIndex + 0.5)*\cellSize}
    \pgfmathsetmacro{\dotY}{(\SubcarrierIndex + 0.5)*\cellSize}

    \fill[myBlue2] (\dotX, \dotY) circle (1.5pt);

    \pgfmathsetmacro{\gridTop}{\numSubcarriers*\cellSize}
    \pgfmathsetmacro{\constellationBottom}{\gridTop + 0.5}
    \pgfmathsetmacro{\constellationCenterX}{\dotX + 1*\cellSize}
    \pgfmathsetmacro{\constellationCY}{\constellationBottom + 0.45}

    \draw[myBlue2, thin, dashed] (\dotX, \dotY) node[below, font=\tiny] {$(\Findex',\Dindex')$} -- (\constellationCenterX, \constellationBottom);

    \def\qamScale{0.2}   
    \def\qamDotSize{1.2pt}
    \pgfmathsetmacro{\qamHalf}{1*\qamScale}

    \node[above, font=\footnotesize] at (\constellationCenterX, \constellationBottom + 0.93)
        {$\constSet_{\constMap(\Findex',\Dindex')}$};

    \draw[->, black!50, thin]
        ({\constellationCenterX - \qamHalf - 0.15}, {\constellationCY})
        --
        ({\constellationCenterX + \qamHalf + 0.15}, {\constellationCY});
    \draw[->, black!50, thin]
        ({\constellationCenterX}, {\constellationBottom + 0.02})
        --
        ({\constellationCenterX}, {\constellationBottom + 0.93});

    \pgfmathsetmacro{\selectedQI}{0}   
    \pgfmathsetmacro{\selectedQQ}{0}   

    \pgfmathsetmacro{\selectedI}{2*\selectedQI - 1}
    \pgfmathsetmacro{\selectedQ}{2*\selectedQQ - 1}

    \foreach \xi in {-1,1} {
        \foreach \yi in {-1,1} {
            \pgfmathsetmacro{\px}{\constellationCenterX + \xi*\qamScale}
            \pgfmathsetmacro{\py}{\constellationCY + \yi*\qamScale}
            \pgfmathsetmacro{\isSelected}{(\xi == \selectedI) && (\yi == \selectedQ) ? 1 : 0}
            \ifnum\isSelected=1
                \fill[myBlue2] (\px, \py) circle (2pt);
                \draw[myBlue2, thin] (\px, \py) circle (3pt) node[below left=1pt, font=\footnotesize] {$\dataMat[\Findex',\Dindex']$};
            \else
                \fill[black] (\px, \py) circle (\qamDotSize);
            \fi
        }
    }

\end{tikzpicture}
    }
    \caption{Illustration of the considered \gls{tfb}. 
    }
    \label{fig:resource_grid}
\end{figure}
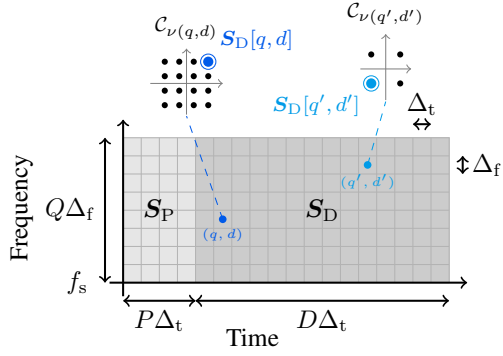

The \gls{ue} and the background are assumed to be stationary during the frame duration $\T \triangleq (\numP+\numD)\Tspacing$, with $\Tspacing$ being the \gls{ofdm} symbol duration. Hence, no Doppler effect is considered.
For the typical \gls{6g} \gls{isac} parameters used in the simulations (see \autoref{tab:parameters}), this yields\footnotemark\ $\Tspacing \approx \Fspacing^{-1} = \SI{22.22}{\micro\second}$ and $\T=(\numP+\numD)\Tspacing \approx \SI{0.8}{\milli\second}$, which is sufficiently short to justify the stationarity assumption within a single localization update interval.
\footnotetext{Neglecting the cyclic prefix of length $\numCP$, since $\Tspacing = \frac{\numF + \numCP}{\numF \Fspacing} \approx \Fspacing^{-1}$ for $\numF \gg \numCP$.}

\subsection{Channel Model and Observations}\label{sec:channel_model}

The channel model assumes \gls{los} propagation and perfect timing-clock synchronization at each node.
Under these assumptions, the channel matrix $\channelMat(\UEpos) \in \C^{\numRX \times \numF}$, whose $(\RXindex,\Findex)$-th element at the $\RXindex$-th node and $\Findex$-th subcarrier reads (see, e.g., \cite{sakhnini_near-field_2022})
\begin{align}
    \channelMat(\UEpos)[\RXindex,\Findex] & = \channel[\RXindex] e^{-\jc 2 \pi (\carrierF + \Findex \Fspacing) 
    \norm{\UEpos - \RXpos{\RXindex}}/c} \\
    & = \underbrace{\channel[\RXindex] e^{-\jc \carrierWn \norm{\UEpos - \RXpos{\RXindex}}}}_{\triangleq \channelcoeff[\RXindex]} 
    \underbrace{e^{-\jc \carrierWn  \norm{\UEpos - \RXpos{\RXindex}} \Findex \frac{\Fspacing}{\carrierF}}}_{\triangleq \stMat(\UEpos)[\RXindex,\Findex]}, \label{eq:channel_model}
\end{align}
where $\channel \in \C^{\numRX \times 1}$ collects the random channel coefficient at all nodes, accounting for propagation losses and random phases at each \gls{rx}; $\channelcoeff \in \C^{\numRX \times 1}$ further incorporates the propagation phases at the reference frequency; and $\stMat(\UEpos) \in \C^{\numRX \times \numF}$ denotes the steering matrix, i.e., the component carrying the localization information.
The carrier-phase information is not exploited for localization and the corresponding term is therefore absorbed into $\channel[\RXindex]$, forming the unknown complex channel gains $\channelcoeff[\RXindex]$.
Furthermore, the coefficients $\channelcoeff[\RXindex]$ are assumed constant across the entire \gls{tfb}, which is therefore contained within a single coherent block \cite{bjornson_massive_2017}—a condition ensured by the stationarity and \gls{los} assumptions.

At the \gls{srx}, pilot and data observations are denoted by $\pilotObs \in \C^{\numRX \times \numF \times \numP}$ and $\dataObs \in \C^{\numRX \times \numF \times \numD}$, respectively, and are given by
\begin{align}
    \pilotObs [\RXindex, \Findex, \Pindex ] &= \channelMat(\UEpos)[\RXindex,\Findex] \pilotMat[\Findex,\Pindex] + \pilotAWGN[\RXindex, \Findex, \Pindex], \label{eq:pilotObs}\\
    \dataObs[\RXindex, \Findex, \Dindex] &= \channelMat(\UEpos)[\RXindex,\Findex] \dataMat[\Findex,\Dindex] + \dataAWGN[\RXindex, \Findex, \Dindex], \label{eq:dataObs}
\end{align}
where $\pilotAWGN \in \C^{\numRX \times \numF \times \numP}$ and $\dataAWGN \in \C^{\numRX \times \numF \times \numD}$ are the \gls{awgn} terms that are independent across all dimensions, where each element follows a complex normal distribution $\CN(0,\noiseVar)$.
For brevity in following derivations, we introduce the time-concatenated quantities $\symbolMat \in \C^{\numF \times \numPD}$ and $\Obs \in \C^{\numRX \times \numF \times \numPD}$, where $\numPD \triangleq \numP + \numD$ denotes the total number of \gls{ofdm} symbols.

\section{JML Framework}\label{sec:JML}

The objective is to estimate the \gls{ue} position $\UEpos$ (i.e., the \gls{poi} here), from both the pilot and data observations.
Since the data symbols $\dataMat$ are unknown at the \gls{srx}, they constitute \glspl{np}, along with the channel coefficients $\channelcoeff$.
Extracting position information from $\pilotObs$ and $\dataObs$ therefore requires eliminating the dependence on these \glspl{np}.
In this paper, we adopt a \gls{jml} approach, in which the uncertainty associated with the \glspl{np} is addressed by jointly estimating them together with the \gls{poi}.
The \gls{jml} estimation problem is expressed as
\begin{equation}\label{eq:JML_formulation}
    \UEposEst^{\JMLo} = \argmax_{\UEposTest} \max_{\dataMatTest \in \constSet_{\constMap}^{\numF\times\numD}} \max_{\channelcoeffTest \in \C^{\numRX \times 1}} \Likelihood(\pilotObs, \dataObs; \UEposTest, \dataMatTest, \channelcoeffTest),
\end{equation}
where $\Likelihood(\pilotObs, \dataObs; \UEposTest, \dataMatTest, \channelcoeffTest)$ denotes the likelihood of observing $(\pilotObs, \dataObs)$ given the candidate \gls{poi} $\UEposTest$ and candidate \glspl{np} $(\dataMatTest, \channelcoeffTest)$.
However, this joint optimization entails a prohibitive computational complexity due to the high-dimensional search space over the \gls{poi}, data symbols, and channel coefficients \cite{chen_joint_2026}.
To address this issue, we adopt a concentrated likelihood approach in which closed-form estimates of the \glspl{np} are derived as explicit functions of the \gls{poi}, and subsequently substituted back into the likelihood function, thereby eliminating their dependence.
To enable closed-form maximization with respect to $\dataMatTest$, this approach requires relaxing the constellation constraint. Specifically, the discrete nature of the data symbols is neglected, and \eqref{eq:JML_formulation} is maximized over $\dataMatTest \in \C^{\numF\times\numD}$ instead of $\dataMatTest \in \constSet_{\constMap}^{\numF\times\numD}$, resulting in
\begin{equation}\label{eq:JMLc_formulation}
    \UEposEst^{\JMLc} = \argmax_{\UEposTest} \max_{\dataMatTest \in \C^{\numF\times\numD}} \max_{\channelcoeffTest \in \C^{\numRX \times 1}} \Likelihood(\pilotObs, \dataObs; \UEposTest, \dataMatTest, \channelcoeffTest),
\end{equation}
where $\UEposEst^{\JMLc}$ denotes the \textit{constellation-relaxed} version of the optimal \gls{jml} solution $\UEposEst^{\JMLo}$.

This section derives the solution to \eqref{eq:JMLc_formulation}, analyzes the resulting estimator structure, and establishes its computational intractability for large \glspl{tfb}.

\subsection{Optimal Constellation-Relaxed JML Solution}

The following \gls{ls} problem is obtained directly from \eqref{eq:JMLc_formulation} by substituting the pilot and data observation models \eqref{eq:pilotObs} and \eqref{eq:dataObs}:
\begin{align}
    & \UEposEst^{\JMLc} =  \argmin_{\UEposTest} \min_{\dataMatTest \in \C^{\numF\times\numD}} \min_{\channelcoeffTest \in \C^{\numRX \times 1}} \nonumber \\
    &  \sum_{\RXindex=0}^{\numRX-1} \sum_{\Findex=0}^{\numF-1} \sum_{\PDindex=0}^{\numPD-1} \abs{\Obs[\RXindex,\Findex,\PDindex] - \channelcoeffTest[\RXindex] \stMat(\UEposTest)[\RXindex,\Findex] \symbolMatTest[\Findex,\PDindex] }^2, \label{eq:JMLc_LS}
\end{align}
where \begingroup\small$\symbolMatTest \triangleq \begin{bmatrix} \pilotMat^T & \dataMatTest^T \end{bmatrix}^T \in \C^{\numF\times\numPD}$\endgroup denotes the column-wise concatenation of the known pilot symbol matrix $\pilotMat$ and the candidate data symbol matrix $\dataMatTest$.

Since the problem of estimating the channel coefficients $\channelcoeff$ is separable in $\RXindex$, a closed-form solution follows via Wirtinger calculus \cite{koor_short_2023} as
\begin{equation}\label{eq:channel_implicit_estimation}
    \channelcoeffEst(\UEposTest, \symbolMatTest)[\RXindex] = \frac{\sum_{\Findex,\PDindex} \stMat^{*}(\UEposTest)[\RXindex,\Findex] \symbolMatTest^{*}[\Findex, \PDindex] \Obs[\RXindex,\Findex,\PDindex]}{\sum_{\Findex,\PDindex} \abs{\stMat(\UEposTest)[\RXindex,\Findex] \symbolMatTest[\Findex, \PDindex]}^2},
\end{equation} 
where the denominator simplifies to the candidate symbol matrix energy $\rgEnergy(\symbolMatTest) \triangleq \sum_{\Findex,\PDindex} |\symbolMatTest[\Findex, \PDindex]|^2$, since $|\stMat(\UEposTest)[\RXindex,\Findex]|^2 = 1 \; \forall \UEposTest$.
These closed-form estimates can then be substituted back into the objective function.

\begin{proposition}\label{prop:JMLc1}
    The \gls{jml} estimation problem \eqref{eq:JMLc_formulation} with $\channelcoeff$ concentrated out reduces to
    \begin{equation}\label{eq:JMLc_LS_data}
     \UEposEst^{\JMLc} =  \argmax_{\UEposTest} \max_{\dataMatTest \in \C^{\numF\times\numD}} \OF(\UEposTest, \symbolMatTest),
    \end{equation}
    where the objective function $\OF$ is given by
    \begin{equation}\label{eq:JMLc_LS_data_OF}
         \OF(\UEposTest, \symbolMatTest) = \frac{ \sum_{\RXindex} \abs{\sum_{\Findex,\PDindex}\Obs^{*}[\RXindex,\Findex,\PDindex] \stMat(\UEposTest)[\RXindex, \Findex] \symbolMatTest[\Findex,\PDindex] }^2}{\sum_{\Findex,\PDindex} \abs{\symbolMatTest[\Findex, \PDindex]}^2}.
    \end{equation}
    
\end{proposition}
\begin{proof}
    See \hyperref[app:JMLc_proof1]{Appendix~\ref*{app:JMLc_proof1}}.
\end{proof}
\begin{remark}\label{rem:spatial_coherence}
    Substituting the channel estimates \eqref{eq:channel_implicit_estimation} in the objective function results in a loss of \underline{spatial coherence} across nodes ($\RXindex$), causing contributions from different receivers to combine incoherently—as reflected by the summation over receiver indices outside the squared norm in \eqref{eq:JMLc_LS_data_OF}.
    This degrades the likelihood function, as discussed in \cite{reniers_localization_2026}.
    This loss arises because $\channelcoeffEst(\UEposTest, \symbolMatTest)[\RXindex]$ depends on the full transmitted frame. 
    However, as shown in the following, a specific separation of the pilot and data contributions enables coherent processing to be recovered for part of the objective function.
    In contrast, \underline{frequency} ($\Findex$) and \underline{time} ($\PDindex$) \underline{coherence} are preserved, as the expression remains conditioned on $\symbolMatTest$ and all receivers share the same transmitted symbol sequence.
\end{remark}

\begin{table*}[t]
    \centering
    \caption{Notation summary for the proposed $\JMLc$ formulation.}
    \label{tab:notations_JMLc}
    \setlength{\extrarowheight}{0.1cm}
    \begin{tabular}{|l|l|p{8.5cm}|}
        \hline 
        Symbol & Dimension & Definition / Interpretation \\
        \hline \hline 
        $\dataVecTest \triangleq \vecc\{\dataMatTest\}$ & $\C^{\numF\numD \times 1}$ & Vectorized candidate data symbols \\
        \hline
        $\pilotEnergy \triangleq \rgEnergy(\pilotMat) =  \sum_{\Findex,\Pindex} |\pilotMat[\Findex, \Pindex]|^2$ & $\R_{+}$ & Pilot component energy\\
        \hline 
        $\channelcoeffEstPilotsUn{\RXindex}(\UEposTest) \triangleq \sum_{\Findex,\Pindex} \pilotObs[\RXindex,\Findex,\Pindex] \stMat^{*}(\UEposTest)[\RXindex, \Findex] \pilotMat^{*}[\Findex,\Pindex]$ & $\C$ & Unnormalized pilot-based channel estimate at receiver $\RXindex$, such that $\channelcoeffEstPilotsUn{\RXindex}(\UEposTest) / \pilotEnergy$ yields the estimate of $\channelcoeff[\RXindex]$ \\
        $\channelcoeffEstPilotsU(\UEposTest) \triangleq \begin{bmatrix} \channelcoeffEstPilotsUn{0}(\UEposTest) & \cdots & \channelcoeffEstPilotsUn{\numRX-1}(\UEposTest) \end{bmatrix}^{\Trans}$ & $\C^{\numRX\times 1}$ & $-$ \\
        $\channelEnergy(\UEposTest) \triangleq \norm{\channelcoeffEstPilotsU(\UEposTest)}^2 = \channelcoeffEstPilotsU^{\Herm}(\UEposTest)\channelcoeffEstPilotsU(\UEposTest)$ & $\R_{+}$ & Squared norm of the unnormalized pilot-based channel estimates across all receiver nodes, such that $\channelEnergy(\UEposTest) / \pilotEnergy^2$ yields the total channel energy \\
        \hline 
        $\dataMatEstcc{\RXindex}(\UEposTest)[\Findex, \Dindex] \triangleq \dataObs[\RXindex,\Findex,\Dindex] \stMat^{*}(\UEposTest)[\RXindex,\Findex]$ & $\dataMatEstcc{\RXindex} \in \C^{\numF \times \numD}$ & Data observations phase-compensated by the steering matrix at candidate position $\UEposTest$, such that $\dataMatEstcc{\RXindex}(\UEposTest)[\Findex, \Dindex] / \channelcoeff[\RXindex]$ yields the soft estimate of data symbol $(\Findex,\Dindex)$ at node $\RXindex$ for candidate position $\UEposTest$ \\
        $\dataVecEstcc{\RXindex}(\UEposTest) \triangleq \vecc\{\dataMatEstcc{\RXindex}(\UEposTest)\}$ & $\C^{\numF\numD \times 1}$ & $-$ \\
        $\dataMatEstccConcat(\UEposTest) \triangleq \begin{bmatrix} \dataVecEstcc{0}(\UEposTest) & \cdots &  \dataVecEstcc{\numRX-1}(\UEposTest) \end{bmatrix} $ & $\C^{\numF\numD \times \numRX}$ & $-$ \\
        \hline
        \makecell[l]{$\dataMatEstccCov(\UEposTest) \triangleq \dataMatEstccConcat(\UEposTest) \dataMatEstccConcat^{\Herm}(\UEposTest) = \sum_{\RXindex} \dataVecEstcc{\RXindex}(\UEposTest) \dataVecEstcc{\RXindex}^{\Herm}(\UEposTest)$} & $\C^{\numF\numD \times \numF\numD}$ & Aggregate Gram matrix of the phase-compensated data observations $\dataVecEstcc{\RXindex}(\UEposTest)$ across all receiver nodes \\
        \hline
        $\dataVecEstccP(\UEposTest) \triangleq \dataMatEstccConcat(\UEposTest) \channelcoeffEstPilotsU^{*}(\UEposTest) = \sum_{\RXindex} \dataVecEstcc{\RXindex} \channelcoeffEstPilotsUn{\RXindex}^{*}(\UEposTest)$ & $\C^{\numF\numD \times 1}$ & Coherent combination of phase-compensated data observations $\dataVecEstcc{\RXindex}(\UEposTest)$ weighted by the pilot-based channel estimates, such that $\dataVecEstccP(\UEposTest) / \pilotEnergy$ yields soft data symbol estimates. \\
        \hline
    \end{tabular}
\end{table*}

To obtain an expression that depends solely on the candidate \gls{poi} $\UEposTest$, the dependency on $\dataMatTest$ (which appears inside $\symbolMatTest$) must be eliminated.
The resulting concentrated estimator is derived in two steps, within which the notations in \autoref{tab:notations_JMLc} are introduced to ease interpretability.
First, the data symbol estimates are expressed as a function of the candidate \gls{poi} $\UEposTest$ by imposing the following stationarity conditions (i.e., setting the Wirtinger derivatives to zero):
\begin{equation}\label{eq:JMLc_stationary_conditions}
    \frac{\partial\OF(\UEposTest,\dataVecTest)}{\partial\dataVecTest } = 0 \quad \text{and} \quad \frac{\partial\OF(\UEposTest,\dataVecTest)}{\partial\dataVecTest^{*}} = 0.
\end{equation}
\begin{proposition}\label{prop:JMLc2}
    The estimate of the data symbol vector, denoted $\dataVecEst \triangleq \dataVecEst(\UEposTest)$, satisfies
    \begin{equation}\label{eq:JMLc_genEVD}
        \dataMatEstccCov(\UEposTest) \dataVecEst + \dataVecEstccP(\UEposTest) = \OF(\UEposTest,\dataVecEst) \dataVecEst,
    \end{equation}
    which constitutes a generalized \gls{evd} problem in the $\dataVecEst$-direction. 
    Consequently, the solution to $\max_{\dataVecTest} \OF(\UEposTest,\dataVecTest)$ is given by the largest generalized eigenvalue of the corresponding augmented formulation uniquely defined by $\dataMatEstccCov(\UEposTest)$ and $\dataVecEstccP(\UEposTest)$.
\end{proposition}
\begin{proof}
    See \hyperref[app:JMLc_proof2]{Appendix~\ref*{app:JMLc_proof2}}.
\end{proof}

Unlike the concentration over $\channelcoeff$, which admits a closed-form solution, the concentration over $\dataMat$ does not, and insteads leads to the generalized \gls{evd} problem in \eqref{eq:JMLc_genEVD}.
The objective function \eqref{eq:JMLc_LS_data_OF} is therefore reformulated, in a second step, as a \textit{generalized Rayleigh Quotient} to cast the problem into a more tractable form, as shown below.
Introducing the augmented candidate vector \begingroup$\dataAugVecTest \triangleq \begin{bmatrix} \dataVecTest^{\Trans} & 1 \end{bmatrix}^{\Trans} \in \C^{(\numF\numD + 1) \times 1}$\endgroup, the objective is defined as $\OFmax(\UEposTest) \triangleq \max_{\dataAugVecTest} \OF(\UEposTest,\dataAugVecTest)$, with
\begin{equation}\label{eq:gen_RQ}
    \OF(\UEposTest,\dataAugVecTest) = \frac{\dataAugVecTest^{\Herm} \begin{bmatrix} \dataMatEstccCov(\UEposTest) &  \dataVecEstccP(\UEposTest) \\  \dataVecEstccP^{\Herm}(\UEposTest) & \channelEnergy(\UEposTest) \end{bmatrix} \dataAugVecTest}{\dataAugVecTest^{\Herm} \begin{bmatrix} \I{\numF\numD} &  \Zero{\numF\numD} \\ \Zero{\numF\numD}^{\Trans} & \pilotEnergy \end{bmatrix} \dataAugVecTest} \triangleq \frac{\dataAugVecTest^{\Herm} \ARQnum(\UEposTest) \dataAugVecTest}{\dataAugVecTest^{\Herm} \ARQdenom \dataAugVecTest},
\end{equation}
directly  obtained from \eqref{eq:OF_dev1}, where $\ARQnum(\UEposTest) \in \C^{(\numF\numD + 1) \times (\numF\numD + 1)}$ is Hermitian and $\ARQdenom \in \C^{(\numF\numD + 1) \times (\numF\numD + 1)}$ is Hermitian positive definite.
Hence, $\OFmax(\UEposTest)$ is obtained as the largest generalized eigenvalue of $\ARQnum(\UEposTest) \dataAugVecTest = \lambda(\UEposTest) \ARQdenom \dataAugVecTest$.
Since $\ARQdenom$ is diagonal with strictly positive entries, the following change of variables reduces this to a standard eigenvalue problem. 
Let $\MatChangeToEVD \triangleq \ARQdenom^{1/2}$ and $\dataAugVecTestChanged \triangleq \MatChangeToEVD \dataAugVecTest = \begin{bmatrix} \dataVecTest^{\Trans} & \sqrt{\pilotEnergy} \end{bmatrix}^{\Trans}$:
\begin{equation}
    \ARQnumChange(\UEposTest) \dataAugVecTestChanged = \lambda(\UEposTest) \dataAugVecTestChanged,
\end{equation}
where $\ARQnumChange(\UEposTest) = \MatChangeToEVD^{-1} \ARQnum(\UEposTest) \MatChangeToEVD^{-1}$.
\begin{proposition}
    The solution to the constellation-relaxed \gls{jml} problem \eqref{eq:JMLc_formulation}, with $\channelcoeff$ and $\dataMat$ concentrated out, is given by
    \begin{equation}\label{eq:JMLc_estimator}
        \UEposEst^{\JMLc} = \argmax_{\UEposTest} \Lambda_{\max} \Big\{ \ARQnumChange(\UEposTest) \Big\},
    \end{equation}
    where $\Lambda_{\max} \left\{ \ARQnumChange(\UEposTest)\right\}$ denotes the largest eigenvalue of matrix
    \begin{equation}\label{eq:JMLc_estimator_matrix}
        \ARQnumChange(\UEposTest) = \begin{bmatrix} \dataMatEstccCov(\UEposTest) &  \frac{\dataVecEstccP(\UEposTest)}{\sqrt{\pilotEnergy}}\\  \frac{\dataVecEstccP^{\Herm}(\UEposTest)}{\sqrt{\pilotEnergy}} & \frac{\channelEnergy(\UEposTest)}{\pilotEnergy} \end{bmatrix}.
    \end{equation}
\end{proposition}

\vspace*{-0.375cm}
\subsection{Structural Analysis and Computational Intractability}

Matrix \eqref{eq:JMLc_estimator_matrix} exhibits an insightful block structure, which is illlustrated in \autoref{fig:ARQnumChange_Illustration} and which we now examine term by term.

\begin{enumerate}
    \item The scalar bottom-right entry $\frac{\channelEnergy(\UEposTest)}{\pilotEnergy}$ corresponds to the pilot-only estimator ($\numD=0$), given by $\UEposEst^{\Pmethod} = \argmax_{\UEposTest} \frac{\channelEnergy(\UEposTest)}{\pilotEnergy} = \argmax_{\UEposTest}\channelEnergy(\UEposTest)$.
    It is equivalently obtained from \eqref{eq:JMLc_LS_data} by restricting to pilot observations and removing the maximization over $\dataMatTest$.
    This estimator is further discussed in Section~\ref{sec:results_baselines} (see \eqref{eq:P_estimator}).
    \item The top-left block $\dataMatEstccCov(\UEposTest)$ captures the data-only contribution ($\numP=0$): $\UEposEst^{\Dmethod} = \argmax_{\UEposTest} \Lambda_{\max} \{ \dataMatEstccCov(\UEposTest)\}$.
    In contrast to the pilot-only case, which concentrates the energy of all $\numF\numP$ pilot observations into a single scalar, the information carried by the data observations is distributed across a $\numF\numD$-dimensional space due to their unknown nature.
    \item Finally, the off-diagonal cross-terms $\dataVecEstccP(\UEposTest)/\sqrt{\pilotEnergy}$ and $\dataVecEstccP^{\Herm}(\UEposTest)/\sqrt{\pilotEnergy}$ reveal the coupling between pilot and data contributions.
    Specifically, they correspond, up to a constant factor, to data estimates whose channel component has been pre-compensated using the pilot-based channel estimates (see \autoref{tab:notations_JMLc}).\\
    As will be shown in the next section, this structure directly motivates the tractable approximation we propose, in which the pilot observations are leveraged to eliminate the channel dependency in the data term.
\end{enumerate}

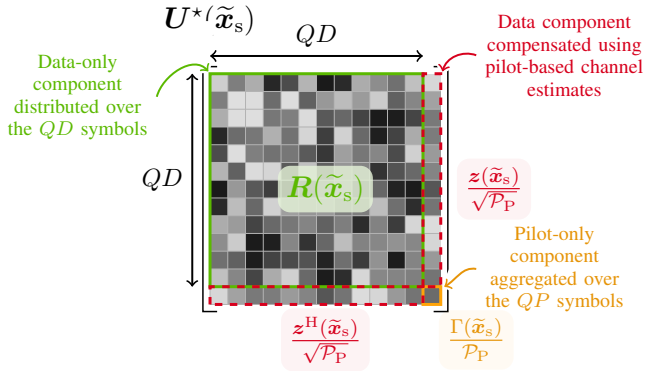
\begin{figure}[t]
    \centering
    \resizebox{1.0\linewidth}{!}{%
        \begin{tikzpicture}
    \def\N{13}
    \def\ps{0.25}
    \def\pad{0.1} 

    \newcommand{\highlightPart}[6]{
        \draw[#5, #6, very thick,] ({(#1)*\ps}, {(#2)*\ps}) rectangle ({(#3)*\ps}, {(#4)*\ps});
    }

    \pgfmathsetseed{1234}

    \foreach \i in {0,...,\numexpr\N-1} {
        \foreach \j in {\i,...,\numexpr\N-1} {

            \pgfmathparse{10 + 80*rnd}
            \let\g\pgfmathresult

            \expandafter\xdef\csname val\i-\j\endcsname{\g}
            \expandafter\xdef\csname val\j-\i\endcsname{\g}
        }
    }

    \foreach \i in {0,...,\numexpr\N-1} {
        \foreach \j in {0,...,\numexpr\N-1} {

            \pgfmathparse{\csname val\i-\j\endcsname}
            \let\g\pgfmathresult

            \fill[black!\g]
                (\i*\ps, {(\N-1-\j)*\ps})
                rectangle ++(\ps,\ps);
        }
    }

    \foreach \i in {0,...,\numexpr\N-1} {
        \draw[black!30] (\i*\ps, 0) -- (\i*\ps, \N*\ps);
    }
    \foreach \j in {0,...,\numexpr\N-1} {
        \draw[black!30] (0,\j*\ps) -- (\N*\ps,\j*\ps);
    }

    \node[very thick, font=\large] at (0,\N*\ps+3*\ps) {$\ARQnumChange(\UEposTest)$};

    \pgfmathsetmacro{\total}{\N*\ps}

    \draw[semithick] (\pad, -\pad) -- (-\pad, -\pad) -- (-\pad, \total+\pad) -- (\pad, \total+\pad);

    \draw[semithick] (\total-\pad, -\pad) -- (\total+\pad, -\pad) -- (\total+\pad, \total+\pad) -- (\total-\pad, \total+\pad);

    \draw[<->, thick] (0, \N*\ps+\ps) -- (\N*\ps-\ps, \N*\ps+\ps) node[midway, above] {$\numF\numD$};       
    \draw[<->, thick] (-\ps, \ps) -- (-\ps, \N*\ps) node[midway, left] {$\numF\numD$};                      

    \highlightPart{0}{1}{\N-1}{\N}{myGreen}{solid}
    \node[thick, myGreen, fill=myGreen!15!white, fill opacity=0.8, text opacity=1, rounded corners=5pt, font=\large] at (0.5*\N*\ps,0.5*\N*\ps) {$\dataMatEstccCov(\UEposTest)$};

    \highlightPart{\N-1}{0}{\N}{1}{myOrange}{solid}
    \node[thick, myOrange!90!black, fill=myOrange, fill opacity=0.05, text opacity=1, rounded corners=5pt, font=\large] at (\N*\ps+2*\ps,-2*\ps) {$\frac{\channelEnergy(\UEposTest)}{\pilotEnergy} $};
    \highlightPart{0}{0}{\N-1}{1}{myRed}{dashed}

    \node[thick, myRed, fill=myRed, fill opacity=0.05, text opacity=1, rounded corners=5pt, font=\large] at(0.5*\N*\ps,-2*\ps) {$\frac{\dataVecEstccP^{\Herm}(\UEposTest)}{\sqrt{\pilotEnergy}}$};
    \highlightPart{\N-1}{1}{\N}{\N}{myRed}{dashed}
    \node[thick, myRed, fill=myRed, fill opacity=0.05, text opacity=1, rounded corners=5pt, font=\large] at(\N*\ps+3*\ps,0.5*\N*\ps) {$\frac{\dataVecEstccP(\UEposTest)}{\sqrt{\pilotEnergy}}$};


    \draw[->, thick, myGreen,preaction={draw, white, line width=8pt}] (-3*\ps, 0.9*\N*\ps) 
    node[anchor=east, align=center,font=\footnotesize] {Data-only \\ component \\ distributed over \\ the $\numF\numD$ symbols}
    to[out=75, in=130] (-0.1*\ps, 1.01*\N*\ps);

    \draw[->, thick, myOrange!90!black,preaction={draw, white, line width=8pt}] (\N*\ps+2*\ps, 2*\ps) 
    node[anchor=west, align=center, font=\footnotesize] {Pilot-only \\ component \\ aggregated over \\ the $\numF\numP$ symbols} 
    to[out=170, in=30] (\N*\ps+0.2*\ps, \ps);

     \draw[->, thick, myRed,preaction={draw, white, line width=8pt}] (\N*\ps+2*\ps, 1.08*\N*\ps) 
    node[anchor=west, align=center, font=\footnotesize] {Data component \\ compensated using \\ pilot-based channel \\  estimates} 
    to[out=160, in=75] (\N*\ps+0.1*\ps, \N*\ps*1.01);

\end{tikzpicture}
    }
    \caption{Illustration of the block structure of matrix $\ARQnumChange(\UEposTest)$, highlighting the pilot-only component, the data-only component, and the data component with pilot-based channel compensation.}
    \label{fig:ARQnumChange_Illustration}
\end{figure}

An insightful analysis consists in examining the behavior of $\Lambda_{\max} \left\{ \ARQnumChange(\UEposTest)\right\}$ as a function of $\UEposTest$, which is expected to increase as $\UEposTest$ approaches the true position $\UEpos$. 
Moreover, the eigenvector associated with the largest eigenvalue, denoted by $\dataAugVecTestChangedMax$, corresponds to an implicit soft estimate of $\dataMat$ when restricted to its first $\numF\numD$ components (i.e., excluding the augmented coefficient). 
\autoref{fig:Lambda_max_illustration} illustrates $\Lambda_{\max} \left\{ \ARQnumChange(\UEposTest)\right\}$ and the corresponding $\dataAugVecTestChangedMax$ for several candidate positions. 
As $\UEposTest \rightarrow \UEpos$, $\Lambda_{\max} \left\{ \ARQnumChange(\UEposTest)\right\}$ increases and the underlying constellation structure progressively emerges, thereby confirming the proposed interpretation.

\begin{figure}[ht]
    \centering
    \includegraphics[width=1.0\linewidth]{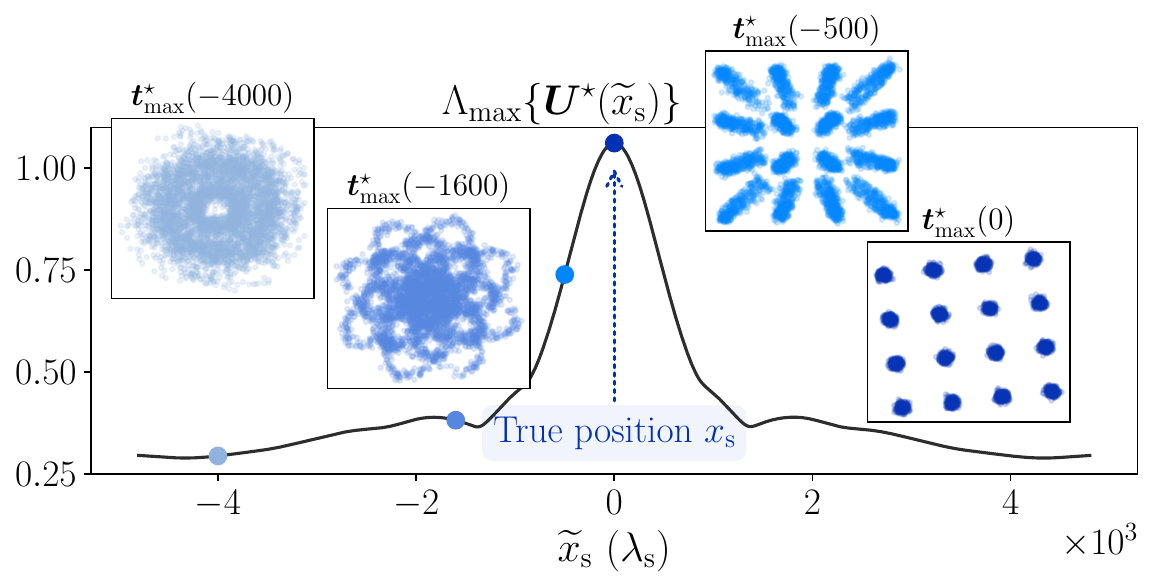}
    \caption{Illustration of $\Lambda_{\max} \left\{ \ARQnumChange(\UEposTest)\right\}$ for a \gls{ue} located at $\UEpos = [0, 0]$. 
    The first $\numF\numD$ elements of $\dataAugVecTestChangedMax$ (corresponding to the implicit estimate of $\dataMat$) are shown for selected candidate positions. 
    The plot represents a cut along $\widetilde{x}_{\mathrm{s}} = \UEposTest[:,0]$, normalized by $\carrierWl$ (here \SI{4.17}{\centi\meter}), and follows the simulation setup from Section~\ref{sec:results}.
    The data constellation is a $16$-\gls{qam} and the \gls{snr} is fixed to \SI{20}{\decibel}; all other parameters are left at their default values (see \autoref{tab:parameters}).}
    \label{fig:Lambda_max_illustration}
\end{figure}

Finally, while being nearly optimal (the only approximation being the constellation relaxation), the estimator $\JMLc$ in \eqref{eq:JMLc_estimator} is computationally intractable for typical \glspl{tfb} used in practice.
As detailed in Section~\ref{sec:Complexity_th}, computing $\Lambda_{\max}\{\ARQnumChange(\UEposTest)\}$ incurs an asymptotic complexity of $\Compl(\numF^3 \numD^3)$ per candidate position, exceeding practical computational limits for typical system parameters, as confirmed by the runtime measurements reported in Section~\ref{sec:Complexity_meas}.
This motivates the development of more tractable approximations.

\section{Practical JML Heuristics}\label{sec:JML_heuristics}

This section first introduces an approximation of the $\JMLc$ estimator.
Next, the resulting estimator is reformulated by means of a \gls{svd}, from which an empirical approximation based on a subset of components is obtained. 
This second approximation preserves the estimation accuracy of the previous estimator, which already outperforms existing baselines, while further reducing the computational complexity.
Finally, two additional scenarios to which the proposed estimators directly apply are discussed.

\subsection{Approximate JML Solution}

Motivated by the structure of the off-diagonal term $\dataVecEstccP(\UEposTest)$ in matrix $\ARQnumChange(\UEposTest)$ (see \autoref{fig:ARQnumChange_Illustration}), we replace the unknown channel coefficients in the data likelihood term with a pilot-based estimate.
This estimate is constructed for each candidate position $\UEposTest$ and denoted by $\channelcoeffEstPilots(\UEposTest)$.
The resulting approximate \gls{jml} estimator is given by
\begin{align}
    & \UEposEst^{\JMLa} = \argmax_{\UEposTest} \max_{\dataMatTest \in \C^{\numF\times\numD}} \max_{\channelcoeffTest\in\C^{\numRX\times1}} \Likelihood(\pilotObs ; \UEposTest, \channelcoeffTest) \nonumber\\
    & \hspace{4cm} \times \Likelihood(\dataObs; \UEposTest, \dataMatTest, \channelcoeffEstPilots(\UEposTest)) \label{eq:JMLa_formulation}\\[0.2em]
    & \textcolor{white}{\UEposEst^{\JMLa}} =  \argmax_{\UEposTest} \underbrace{\max_{\channelcoeffTest\in\C^{\numRX\times1}} \Likelihood(\pilotObs ; \UEposTest, \channelcoeffTest)}_{\text{Pilot term}} \nonumber \\
    & \hspace{2cm} \times \underbrace{\max_{\dataMatTest \in \C^{\numF\times\numD}} \Likelihood(\dataObs; \UEposTest, \dataMatTest, \channelcoeffEstPilots(\UEposTest))}_{\text{Data term}},\label{eq:JMLa_formulation2}
\end{align}
where \eqref{eq:JMLa_formulation} follows from \eqref{eq:JMLc_formulation} by injecting the approximation 
\begin{equation}\label{eq:approximation}
    \Likelihood(\dataObs; \UEposTest, \dataMatTest, \channelcoeffTest) \approx \Likelihood(\dataObs; \UEposTest, \dataMatTest, \channelcoeffEstPilots(\UEposTest)).
\end{equation}

This approximation eliminates the channel dependence in the data likelihood term, which now depends only on the single \gls{np} term $\dataMatTest$ and can therefore be handled efficiently, as highlighted below.
The pilot-based channel coefficient estimate is obtained as
\begin{equation}\label{eq:channelcoeffEstPilots_ML}
    \channelcoeffEstPilots(\UEposTest)[\RXindex] = \argmax_{\channelcoeffTest[\RXindex] \in \C} \Likelihood(\pilotObs[\RXindex,:,:]; \channelcoeffTest[\RXindex], \UEposTest),
\end{equation}
which is separable across the $\numRX$ receivers.
Using Wirtinger calculus \cite{koor_short_2023}, analogously to \eqref{eq:channel_implicit_estimation}, this yields
\begin{align}
    \channelcoeffEstPilots(\UEposTest)[\RXindex] & = \frac{\sum_{\Findex,\Pindex=0,0}^{\numF-1,\numP-1} \stMat^{*}(\UEposTest)[\RXindex,\Findex] \pilotMat^{*}[\Findex,\Pindex]  \pilotObs[\RXindex, \Findex, \Pindex]}{\sum_{\Findex,\Pindex=0,0}^{\numF-1,\numP-1}\abs{\stMat(\UEposTest)[\RXindex,\Findex] \pilotMat[\Findex,\Pindex]}^2} \label{eq:channelcoeffEstPilotsFirst} \\
    & = \frac{1}{\pilotEnergy} \sum_{\Findex=0}^{\numF-1} \stMat^{*}(\UEposTest)[\RXindex,\Findex] \pilotObsEq[\RXindex,\Findex], \label{eq:channelcoeffEstPilots}
\end{align}
where the denominator in \eqref{eq:channelcoeffEstPilotsFirst} reduces to the pilot energy $\pilotEnergy$ since $\abs{\stMat(\UEposTest)[\RXindex,\Findex]}^2 = 1$ for all $\UEposTest$, and where the ``symbol-equalized'' pilot observations $\pilotObsEq \in \C^{\numRX \times \numF}$ are defined as
\begin{equation}\label{eq:pilotObsEq}
    \pilotObsEq[\RXindex,\Findex] \triangleq \sum_{\Pindex=0}^{\numP-1} \pilotMat^{*}[\Findex,\Pindex]  \pilotObs[\RXindex, \Findex, \Pindex],
\end{equation}
as this quantity will be useful in subsequent expressions.
For clarity in the following, we also define $\channelConstruct(\UEposTest) \in \C^{\numRX \times \numF}$ as the channel estimate for candidate position $\UEposTest$, obtained by combining the pilot-based estimate of $\channelcoeff[\RXindex]$ with the steering matrix defined in \eqref{eq:channel_model}:
\begin{equation}
    \channelConstruct(\UEposTest)[\RXindex,\Findex] = \channelcoeffEstPilots(\UEposTest)[\RXindex] \stMat(\UEposTest)[\RXindex,\Findex].
\end{equation}

\begin{proposition}\label{prop:JMLa}
    The estimator $\UEposEst^{\JMLa}$ defined in \eqref{eq:JMLa_formulation2} is given by
    \begin{align}
        & \UEposEst^{\JMLa} =  \argmax_{\UEposTest} \frac{1}{\pilotEnergy} \sum_{\RXindex=0}^{\numRX-1} \abs{\sum_{\Findex=0}^{\numF-1} \pilotObsEq^{*}[\RXindex,\Findex] \stMat(\UEposTest)[\RXindex,\Findex]}^2 \nonumber \\
        + & \frac{1}{\channelConstructEnergy(\UEposTest)} \sum_{\Findex=0}^{\numF-1} \sum_{\Dindex=0}^{\numD-1} \abs{\sum_{\RXindex=0}^{\numRX-1} \dataObs^{*}[\RXindex,\Findex,\Dindex] \channelConstruct(\UEposTest)[\RXindex,\Findex]}^2, \label{eq:JMLa_estimator}
    \end{align}
    where $\channelConstructEnergy(\UEposTest) \triangleq \norm{\channelcoeffEstPilots(\UEposTest)}^2.$
\end{proposition}
\begin{proof}
    See \hyperref[app:JMLa_proof]{Appendix~\ref*{app:JMLa_proof}}.
\end{proof}

\begin{remark}
    In \eqref{eq:JMLa_estimator}, the contributions of different receivers combine incoherently in the pilot term (as in \autoref{rem:spatial_coherence}, where the data dependency has not yet been eliminated), while the phase correction provided by $\channelcoeffEstPilots(\UEposTest)$ enables coherent combining in the data term (subject to pilot-based estimation quality).
    This recovery of spatial coherence in the data term is possible because the channel coefficients $\channelcoeff$ remain constant across the full \gls{tfb}, and the pilot-based estimate—derived from observations independent of the data symbols—effectively serves as a phase reference.
    However, since the transmitted data symbols vary randomly across subcarriers and time instances, coherence in both frequency and time is lost in the second term, as reflected by their summations appearing outside the squared norm.
\end{remark}

\subsection{Low-rank Approximation}\label{sec:JML_fast}
Since the steering matrix $\stMat(\UEposTest) \in \C^{\numRX\times\numF}$ is constant across the time dimension, the data observations carry no additional temporal structure. 
Consequently, the summation over $\Dindex$ in \eqref{eq:JMLa_estimator} can be compressed via an \gls{svd}, and the estimator reformulated using the equivalence established in \autoref{prop:JMLfast}.
\begin{proposition}\label{prop:JMLfast}
    Denote the observations on subcarrier $\Findex$ as $\dataObsMat{\Findex} \triangleq \dataObs[:,\Findex,:] \in \C^{\numRX\times\numD}$ and let $\dataObsMat{\Findex} = \SVDMatLeft{\Findex} \SVDMatVals{\Findex} \SVDMatRight{\Findex}^{\Herm}$ be its \gls{svd}, where $\SVDMatLeft{\Findex} \in \C^{\numRX\times\numRX}$ (resp. $\SVDMatRight{\Findex} \in \C^{\numD\times\numD}$) is a unitary matrix whose columns are the left (resp. right) singular vectors of $\dataObsMat{\Findex}$, and $\SVDMatVals{\Findex} \in \C^{\numRX\times\numD}$ is a rectangular matrix whose diagonal entries are the singular values of $\dataObsMat{\Findex}$ and all off-diagonal entries are zero.
    The data contribution of \eqref{eq:JMLa_estimator} can be equivalently rewritten as
    \begin{align}
        & \frac{1}{\channelConstructEnergy(\UEposTest)} \sum_{\Findex=0}^{\numF-1} \sum_{\Dindex=0}^{\numD-1} \abs{\sum_{\RXindex=0}^{\numRX-1} \dataObs^{*}[\RXindex,\Findex,\Dindex] \channelConstruct(\UEposTest)[\RXindex,\Findex]}^2 = \nonumber \\
         & \frac{1}{\channelConstructEnergy(\UEposTest)} \sum_{\Findex=0}^{\numF-1} \sum_{\SVindex=0}^{\numSV-1} \abs{\SVDValVec{\Findex}[\SVindex]}^2 \abs{\sum_{\RXindex=0}^{\numRX-1} \SVDMatLeft{\Findex}^{*}[\RXindex,\SVindex] \channelConstruct(\UEposTest)[\RXindex,\Findex]}^2, \label{eq:JMLfast_equivalence}
    \end{align}
    where $\SVDValVec{\Findex} \triangleq \begin{bmatrix} \SVDMatVals{\Findex}[0,0] & \cdots & \SVDMatVals{\Findex}[\numSV-1,\numSV-1] \end{bmatrix}^{\Trans} \in \R^{\numSV\times1}_{+}$ and $\numSV = \min(\numRX,\numD)$ denotes the number of singular values.
\end{proposition}
\begin{proof}
    See \hyperref[app:JMLfast_proof]{Appendix~\ref*{app:JMLfast_proof}}.
\end{proof}
When $\numD > \numRX$, we have $\numSV < \numD$, and computational savings may already be achieved if the \gls{svd} overhead is smaller than the gain obtained by summing \eqref{eq:JMLfast_equivalence} over $\numSV$ rather than $\numD$ terms.
Furthermore, \autoref{fig:SV_relevant_illustration} shows that only the first (i.e., largest) singular value contributes significantly as the \gls{snr} increases. 
This empirical observation is consistent with the single-target assumption (i.e., a single \gls{ue} transmits a communication signal over the observed \gls{tfb}) and hence motivates the following approximation, in which the summation over $\SVindex$ is truncated to its first term:
\begin{align}
    & \UEposEst^{\JMLfast} = \argmax_{\UEposTest} \frac{1}{\pilotEnergy} \sum_{\RXindex=0}^{\numRX-1} \abs{\sum_{\Findex=0}^{\numF-1} \pilotObsEq^{*}[\RXindex,\Findex] \stMat(\UEposTest)[\RXindex,\Findex]}^2 \nonumber \\
    & + \frac{1}{\channelConstructEnergy(\UEposTest)} \sum_{\Findex=0}^{\numF-1} \abs{\SVDValVec{\Findex}[0]}^2 \abs{\sum_{\RXindex=0}^{\numRX-1} \SVDMatLeft{\Findex}^{*}[\RXindex,0] \channelConstruct(\UEposTest)[\RXindex,\Findex]}^2. \label{eq:JMLfast_estimator}
\end{align}
While the validity of this approximation at low \gls{snr} is not theoretically guaranteed, Section~\ref{sec:results} demonstrates empirically that it yields the same performance as $\JMLa$ over a large \gls{snr} range, while reducing its computational 
complexity (detailed in Section~\ref{sec:Complexity}) to a similar order as the considered \gls{dd} baselines.

\begin{figure}[ht]
    \centering
    \includegraphics[width=1.0\linewidth]{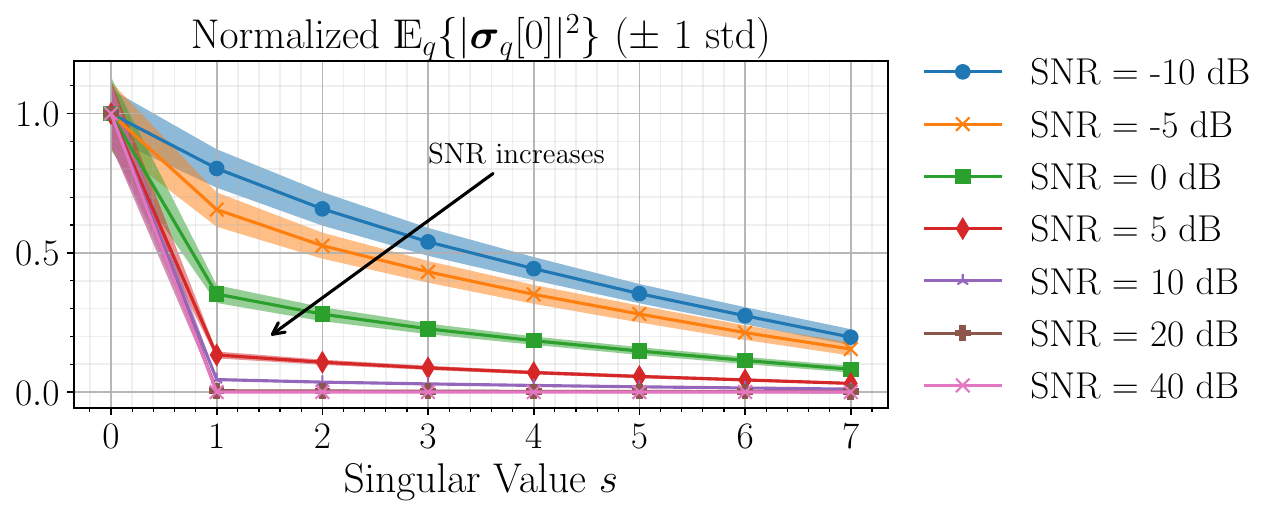}
    \caption{Normalized singular values $\E_{\Findex} \{\abs{\SVDValVec{\Findex}[\SVindex]}^2\}$ ($\pm$ one standard deviation indicated by the shaded area) as a function of \gls{snr}.
    This follows the simulation setup from Section~\ref{sec:results}, with all parameters set to their default values (see \autoref{tab:parameters}).
    The dominant contribution of the first singular value ($\SVindex = 0$) is apparent, supporting the truncation approximation in \eqref{eq:JMLfast_estimator}.}
    \label{fig:SV_relevant_illustration}
\end{figure}

\subsection{Extension to Related Sensing Scenarios}\label{sec:other_scenarios}

While the source localization scenario is considered as a representative instance, the proposed framework and estimators apply directly to other configurations, two of which are illustrated in \autoref{fig:other_scenarios}.

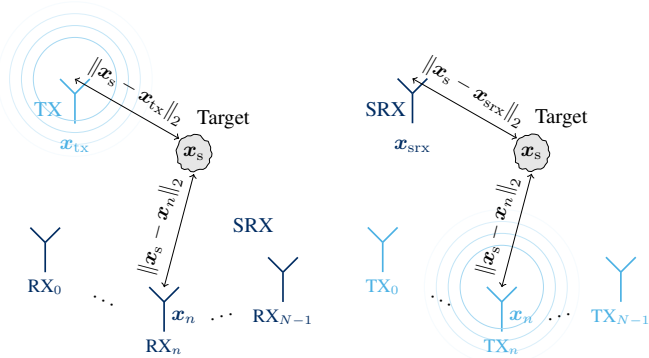
\begin{figure}[ht]
    \centering
    \hspace{-1.5cm}
    \resizebox{0.49\linewidth}{!}{
        \trimbox{1.25cm 0cm 1cm 0cm}{%
        \def\SRXCOLOR{UCLouvainDarkBlue}
\def\TXCOLOR{UCLouvainMediumBlue}
\def\targetCOLOR{black}

\newcommand{\antenna}[3][1]{%
    \draw[thick,#3] #2 -- ++(0,0.5*#1);
    \draw[thick,#3] #2 ++(0,0.5*#1) -- ++(-0.25*#1,0.25*#1);
    \draw[thick,#3] #2 ++(0,0.5*#1) -- ++(0.25*#1,0.25*#1);
}

\pgfmathsetseed{42} 
\newcommand\irregularcircle[2]{%
  \pgfextra{\pgfmathsetmacro{\len}{(#1) + rand*(#2)}}
  +(0:\len cm)
  \foreach \a in {10,20,...,350}{
    \pgfextra{\pgfmathsetmacro{\len}{(#1) + rand*(#2)}}
    -- +(\a:\len cm)
  } -- cycle
}

\def\ueX{-0.5}
\def\ueY{0}
\def\txX{-2.5}
\def\txY{0.5}
\def\ueradius{0.1}
\def\antennaheight{0.75}    
\def\antennapositions{(-3,-2), (-1,-3), (1,-2.5)}
\def\cpuY{-4.5}             
\def\cpuX{0.5}              


\begin{tikzpicture}
    \coordinate (tx) at (\txX,\txY);
    \antenna[1.0]{(tx)}{\TXCOLOR}
    \draw[\TXCOLOR!50!white] ($(tx)+(0,\antennaheight)$) circle (0.7);
    \draw[\TXCOLOR!35!white] ($(tx)+(0,\antennaheight)$) circle (0.9);
    \draw[\TXCOLOR!20!white] ($(tx)+(0,\antennaheight)$) circle (1.1);
    \draw[\TXCOLOR!5!white] ($(tx)+(0,\antennaheight)$) circle (1.3);
    \node[above,\TXCOLOR] at ($(tx)+(-0.45,0)$) {TX};
    \node[below,\TXCOLOR] at ($(tx)+(0,-0.15)$){$\boldsymbol{x}_{\mathrm{tx}}$};

    \coordinate (ue) at (\ueX,\ueY);
    \fill[\targetCOLOR] (ue) circle (\ueradius);
    \draw[\targetCOLOR, rounded corners={2*\ueradius}, fill=\targetCOLOR!10!white] (ue) \irregularcircle{3*\ueradius}{\ueradius/4};
    \node[above,\targetCOLOR] at ($(ue)+(0.5,0.3)$) {Target};
    \node[\targetCOLOR] at ($(ue)+(0,0)$) {$\UEpos$};

    \foreach [count=\i from 0] \pos in \antennapositions {
        \antenna[1.0]{\pos}{\SRXCOLOR}
        \path \pos coordinate (ant\i);
    }

    \node[below,\SRXCOLOR] at (ant0) {\small RX$_{0}$};
    \node[below,\SRXCOLOR] at (ant1) {\small RX$_{n}$};
    \node[below,\SRXCOLOR] at (ant2) {\small RX$_{N-1}$};

    \draw[white,opacity=0] ($(ant1)+(0,\antennaheight)$) circle (0.7);
    \draw[white,opacity=0] ($(ant1)+(0,\antennaheight)$) circle (0.9);
    \draw[white,opacity=0] ($(ant1)+(0,\antennaheight)$) circle (1.1);
    \draw[white,opacity=0] ($(ant1)+(0,\antennaheight)$) circle (1.3);

    \path let \p1=(ant1), \p2=(ant2),
              \n1={atan2(\y2-\y1,\x2-\x1)} in
        node[rotate=\n1]
        at ($(ant1)!0.5!(ant2)$) {$\cdots$};

    \path let \p1=(ant0), \p2=(ant1),
              \n1={atan2(\y2-\y1,\x2-\x1)} in
        node[rotate=\n1]
        at ($(ant0)!0.5!(ant1)$) {$\cdots$};

    \node[right,\SRXCOLOR] at ($(ant1)+(0,+0.25)$)
        {$\RXpos{\RXindex}$};

    \draw[<->,sloped]
        ($(ant1)+(0,\antennaheight)$)
        --
        ($(ue)+(0*\ueradius,-3.5*\ueradius)$)
        node[
            pos=0.5,
            align=center,
            above=0.02cm,
            fill=white,
            fill opacity=0.0,
            text opacity=1
        ]
        {$\norm{\UEpos-\RXpos{n}}$};

    \draw[<->,sloped]
        ($(tx)+(0,\antennaheight)$)
        --
        ($(ue)+(-2.5*\ueradius,2.5*\ueradius)$)
        node[
            pos=0.5,
            align=center,
            above=0.02cm,
            fill=white,
            fill opacity=0.0,
            text opacity=1
        ]
        {$\norm{\UEpos-\boldsymbol{x}_{\mathrm{tx}}}$};

    \node[\SRXCOLOR] at (0.5,-1.2) {SRX};

\end{tikzpicture}
        }
    }
    \resizebox{0.49\linewidth}{!}{
        \trimbox{1.25cm 0cm 1cm 0cm}{%
        \def\SRXCOLOR{UCLouvainMediumBlue}
\def\TXCOLOR{UCLouvainDarkBlue}
\def\targetCOLOR{black}

\newcommand{\antenna}[3][1]{%
    \draw[thick,#3] #2 -- ++(0,0.5*#1);
    \draw[thick,#3] #2 ++(0,0.5*#1) -- ++(-0.25*#1,0.25*#1);
    \draw[thick,#3] #2 ++(0,0.5*#1) -- ++(0.25*#1,0.25*#1);
}

\pgfmathsetseed{42} 
\newcommand\irregularcircle[2]{%
  \pgfextra{\pgfmathsetmacro{\len}{(#1) + rand*(#2)}}
  +(0:\len cm)
  \foreach \a in {10,20,...,350}{
    \pgfextra{\pgfmathsetmacro{\len}{(#1) + rand*(#2)}}
    -- +(\a:\len cm)
  } -- cycle
}

\def\ueX{-0.5}
\def\ueY{0}
\def\txX{-2.5}
\def\txY{0.5}
\def\ueradius{0.1}
\def\antennaheight{0.75}    
\def\antennapositions{(-3,-2), (-1,-3), (1,-2.5)}
\def\cpuY{-4.5}             
\def\cpuX{0.5}              


\begin{tikzpicture}

    \coordinate (tx) at (\txX,\txY);
    \antenna[1.0]{(tx)}{\TXCOLOR}
    \draw[white,opacity=0] ($(tx)+(0,\antennaheight)$) circle (0.7);
    \draw[white,opacity=0] ($(tx)+(0,\antennaheight)$) circle (0.9);
    \draw[white,opacity=0] ($(tx)+(0,\antennaheight)$) circle (1.1);
    \draw[white,opacity=0] ($(tx)+(0,\antennaheight)$) circle (1.3);

    \node[above,\TXCOLOR] at ($(tx)+(-0.45,0)$) {SRX};
    \node[below,\TXCOLOR] at ($(tx)+(0,-0.15)$){$\boldsymbol{x}_{\mathrm{srx}}$};

    \coordinate (ue) at (\ueX,\ueY);
    \fill[\targetCOLOR] (ue) circle (\ueradius);
    \draw[\targetCOLOR, rounded corners={2*\ueradius}, fill=\targetCOLOR!10!white] (ue) \irregularcircle{3*\ueradius}{\ueradius/4};
    \node[above,\targetCOLOR] at ($(ue)+(0.5,0.3)$) {Target};
    \node[\targetCOLOR] at ($(ue)+(0,0)$) {$\UEpos$};

    \foreach [count=\i from 0] \pos in \antennapositions {
        \antenna[1.0]{\pos}{\SRXCOLOR}
        \path \pos coordinate (ant\i);
        
    }

    \draw[\SRXCOLOR!50!white] ($(ant1)+(0,\antennaheight)$) circle (0.7);
    \draw[\SRXCOLOR!35!white] ($(ant1)+(0,\antennaheight)$) circle (0.9);
    \draw[\SRXCOLOR!20!white] ($(ant1)+(0,\antennaheight)$) circle (1.1);
    \draw[\SRXCOLOR!5!white] ($(ant1)+(0,\antennaheight)$) circle (1.3);

    \node[below,\SRXCOLOR] at (ant0) {\small TX$_{0}$};
    \node[below,\SRXCOLOR] at (ant1) {\small TX$_{n}$};
    \node[below,\SRXCOLOR] at (ant2) {\small TX$_{N-1}$};

    \path let \p1=(ant1), \p2=(ant2),
              \n1={atan2(\y2-\y1,\x2-\x1)} in
        node[rotate=\n1]
        at ($(ant1)!0.5!(ant2)$) {$\cdots$};

    \path let \p1=(ant0), \p2=(ant1),
              \n1={atan2(\y2-\y1,\x2-\x1)} in
        node[rotate=\n1]
        at ($(ant0)!0.5!(ant1)$) {$\cdots$};

    \node[right,\SRXCOLOR] at ($(ant1)+(0,+0.25)$)
        {$\RXpos{\RXindex}$};

    \draw[<->,sloped]
        ($(ant1)+(0,\antennaheight)$)
        --
        ($(ue)+(0*\ueradius,-3.5*\ueradius)$)
        node[
            pos=0.5,
            align=center,
            above=0.02cm,
            fill=white,
            fill opacity=0.0,
            text opacity=1
        ]
        {$\norm{\UEpos-\RXpos{n}}$};

    \draw[<->,sloped]
        ($(tx)+(0,\antennaheight)$)
        --
        ($(ue)+(-2.5*\ueradius,2.5*\ueradius)$)
        node[
            pos=0.5,
            align=center,
            above=0.02cm,
            fill=white,
            fill opacity=0.0,
            text opacity=1
        ]
        {$\norm{\UEpos-\boldsymbol{x}_{\mathrm{srx}}}$};

\end{tikzpicture}
        }
    }
    \caption{Illustration of two additional sensing scenarios to which the proposed framework directly applies: multistatic sensing (left) and distributed transmitters with a single-antenna \gls{srx} (right).}
    \label{fig:other_scenarios}
\end{figure}

\textbf{Multistatic opportunistic sensing} (left panel): a single-antenna \gls{tx} at known position $\boldsymbol{x}_{\mathrm{tx}}$ transmits an \gls{ofdm} communication signal to an intended receiver. 
The \gls{srx} exploits the signal echoes reflected by a target (the main reflector) at position $\UEpos$—which may be the intended \gls{ue}, another device, or any passive scatterer. 
Under the single-bounce assumption, the channel model retains the form $\channelMat(\UEpos)[\RXindex,\Findex] = \channelcoeff[\RXindex] \stMat(\UEpos)[\RXindex,\Findex]$, where $\channelcoeff[\RXindex]$ now encapsulates path loss, radar cross-section, and propagation phase. 
The steering matrix is adapted to account for the bistatic range at each node $\RXindex$:
\begin{equation}\label{eq:stMat_2}
    \stMat(\UEpos)[\RXindex,\Findex] = e^{-\jc \carrierWn ( \norm{\UEpos - \boldsymbol{x}_{\mathrm{tx}}} + \norm{\UEpos - \RXpos{\RXindex}} ) \Findex \frac{\Fspacing}{\carrierF}}.
\end{equation}

\textbf{Distributed transmitter} (right panel): $\numRX$ single-antenna \gls{tx} nodes transmit the same communication signal, whose reflections from the target are collected at a single-antenna \gls{srx} at position $\boldsymbol{x}_{\mathrm{srx}}$. Provided the transmitted signals are mutually orthogonal (e.g., via time or frequency division), the \gls{srx} can separate their contributions, and the channel model remains identical in form. 
Indexing the \gls{tx} nodes by $\RXindex$, the steering matrix becomes
\begin{equation}\label{eq:stMat_3}
    \stMat(\UEpos)[\RXindex,\Findex] = e^{-\jc \carrierWn ( \norm{\UEpos - \RXpos{\RXindex}} + \norm{\UEpos - \boldsymbol{x}_{\mathrm{srx}}} ) \Findex \frac{\Fspacing}{\carrierF}}.
\end{equation}

In both cases, all proposed estimators remain valid; only the steering matrix definition needs to be substituted accordingly with \eqref{eq:stMat_2} or \eqref{eq:stMat_3}. 
The hybrid \gls{toa}-\gls{tdoa} interpretation revealed in Section~\ref{sec:JML_ToA_TDoA} also carries over, with range circle loci replaced by ellipsoidal bistatic range loci.

Finally, extension to multi-target scenarios is straightforward for the passive sensing cases above, where multiple ($K$) targets contribute additively to the channel model. 
The proposed objective functions could be directly reused\footnotemark, and existing multi-target techniques—e.g., successive interference cancellation \cite{jeong_interference_2025, lee_joint_2020}, orthogonal matching pursuit, least squares \cite{pesavento_three_2023} or G-iMUSIC \cite{willame_g-imusic_2026}—applied to extract the $K$ target positions, although this extension is outside the scope of this paper and left for future work.
Extension to multi-\gls{ue} source localization is less relevant in practice, as communication \glspl{ue} in uplink are typically separated via orthogonal multiple access schemes.

\footnotetext{For $\JMLfast$, the truncation is extended to the first $K$ singular values rather than only the first one.}
\section{Numerical Results}\label{sec:results}

This section presents \gls{mc} simulation results assessing the localization performance of the proposed estimators against existing baselines. 
Furthermore, the objective functions are visualized and analyzed theoretically, jointly providing geometric insight into the behavior of the proposed methods and establishing the hybrid \gls{toa}-\gls{tdoa} interpretation.

\subsection{Baselines}\label{sec:results_baselines}

All methods follow a two-stage localization procedure: an initial \textit{coarse estimate} obtained via a grid search over $\Ngrid$ discrete points spanning the scene, followed by a \textit{refinement} stage employing the Nelder-Mead optimization algorithm \cite{nelder_simplex_1965}. 
The following baselines are considered for comparison.

\subsubsection{\texorpdfstring{\textnormal{\Pmethod}}{Pilot}}
The simplest estimator relies solely on the pilot component, discarding the data part entirely.
This is the traditional approach employed in current systems, obtained by retaining only the pilot term of \eqref{eq:JMLa_estimator} (equivalently, by setting $\numD=0$ in \eqref{eq:JMLc_estimator}) and discarding constant coefficients:
\begin{align}
    \UEposEst^{\Pmethod} = \argmax_{\UEposTest} & 
    \sum_{\RXindex=0}^{\numRX-1} \Bigg|\sum_{\Findex=0}^{\numF-1} \pilotObsEq^{*}[\RXindex,\Findex] \stMat(\UEposTest)[\RXindex,\Findex]\Bigg|^2.
    \label{eq:P_estimator}
\end{align}
Note that this estimator originates from the \gls{jml} philosophy, restricted to pilot observations and a single \gls{np} vector term $\channelcoeff$. 
Interestingly, applying the \gls{mml} philosophy with a circularly symmetric complex Gaussian prior on $\channelcoeff$ yields the same estimator (see Equation~(19) in \cite{reniers_localization_2026}).

\subsubsection{\texorpdfstring{\textnormal{\PDmethod}}{Genie}}
This estimator represents the best theoretically achievable performance, established under the assumption that the data symbols are perfectly known (or perfectly demodulated) at the \gls{srx}.
Under this assumption, all $\numP+\numD$ symbols are treated as pilots, and the pilot-equalized observations 
$\pilotObsEq$ in \eqref{eq:P_estimator} are replaced by their full-frame counterpart $\ObsEq \in \C^{\numRX \times \numF}$, defined as
\begin{equation}
    \ObsEq[\RXindex,\Findex] = \pilotObsEq[\RXindex,\Findex] + \sum_{\Dindex=0}^{\numD-1} \dataMat^{*}[\Findex,\Dindex] \dataObs[\RXindex,\Findex,\Dindex].
\end{equation}

\subsubsection{\texorpdfstring{\textnormal{\textsc{Decision-Directed}}}{Decision-Directed}}
In the \gls{dd} approach, the \gls{ue} position is estimated by leveraging the \textit{data symbol estimates as additional pilots}.
This method follows the structure of $\PDmethod$, substituting the known $\dataMat$ with its estimate, and thus suffers from potential demodulation errors.
Since the \gls{srx} is distributed, demodulation can be carried out in either a \textit{distributed} or \textit{centralized} manner.
In the distributed case, each of the $\numRX$ nodes independently estimates its $\numF$ channel coefficients from the pilot observations and performs local data demodulation, producing $\numRX$ potentially different sequences collected in $\dataTensEst \in \C^{\numRX \times \numF \times \numD}$.
In the centralized case, each node forwards its data observations to the \gls{cpu}, which estimates the channel coefficients and performs data demodulation to obtain a single sequence $\dataMatEst \in \C^{\numF \times \numD}$.
In both cases, channel and data estimation rely on the \gls{lmmse} criterion\footnote{A \gls{zf}-based estimator was also considered but yields results very similar to, though marginally worse than, the \gls{lmmse}-based approach; it is therefore omitted for clarity.}.
The \textit{distributed} \gls{dd} estimator takes the form
\begin{align}
    & \UEposEst^{\mathrm{DD}} = \argmax_{\UEposTest} \sum_{\RXindex=0}^{\numRX-1} \Bigg| \sum_{\Findex=0}^{\numF-1} \ObsEqDD^{*}[\RXindex,\Findex] \stMat(\UEposTest)[\RXindex,\Findex]\Bigg|^2 , \, \text{with} \label{eq:DD_estimator} \\
    & \ObsEqDD[\RXindex,\Findex] \triangleq \pilotObsEq[\RXindex,\Findex] + \sum_{\Dindex=0}^{\numD-1} \dataTensEst^{*}[\RXindex,\Findex,\Dindex]  \dataObs[\RXindex,\Findex,\Dindex], \label{eq:DD_estimator_ObsEq}
\end{align}
while the \textit{centralized} variant is obtained by replacing $\dataTensEst[\RXindex,\Findex,\Dindex]$ with $\dataMatEst[\Findex,\Dindex]$ in \eqref{eq:DD_estimator_ObsEq}.
Data estimates can either be projected onto the nearest constellation point, referred to as \textit{hard} decisions (\gls{hdd}), or left in the complex plane without mapping, referred to as \textit{soft} decisions (\gls{sdd}).
The \gls{sdd} case is shown to consistently underperform \gls{hdd} in \cite{reniers_localization_2026}; it is therefore not considered further, and \gls{hdd} and \gls{dd} are used interchangeably in the following.
Note that these \gls{dd} methods operate directly on symbols without accounting for error-correction coding, as we assume the passive system either lacks access to coding information or prefers to avoid performing such computationally intensive decoding.

\subsubsection{\texorpdfstring{\textnormal{\textsc{Marginal Maximum Likelihood}}}{Marginal Maximum Likelihood}}
As discussed in Section~\ref{sec:introduction}, an alternative philosophy to eliminate \glspl{np} dependency consists in marginalizing the conditional likelihood, i.e., integrating over the prior distribution of the \glspl{np}.
The \gls{mml} estimation problem is expressed as
\begin{equation}\label{eq:MML_formulation}
     \UEposEst^{\MMLo} =  \argmax_{\UEposTest} \underbrace{\E_{\dataMat,\channelcoeff} 
    \left\{ \Likelihood(\pilotObs, \dataObs | \dataMat, \channelcoeff; \UEposTest) \right\}}_{\Likelihood(\pilotObs, \dataObs ; \UEposTest) }
\end{equation}
where $\Likelihood(\pilotObs, \dataObs | \dataMat, \channelcoeff ; \UEposTest)$ denotes the likelihood of observing $(\pilotObs, \dataObs)$ conditioned on the stochastic \glspl{np} $(\dataMat, \channelcoeff)$ and given the candidate \gls{poi} $\UEposTest$.
In \cite{reniers_localization_2026}, we derived the optimal estimator based on \eqref{eq:MML_formulation}, which was shown to be computationally intractable (even for very small \glspl{tfb}).
We then introduced a tractable approximation $\MMLa$ based on the same approximation as in \eqref{eq:approximation}.
Furthermore, we introduced a novel analytical acceleration $\MMLfast$ for \gls{qam} constellations by leveraging their geometrical structure, reducing the complexity from $\Compl(\constSize)$ to $\Compl(\sqrt{\constSize})$, where $\constSize$ denotes the constellation size, while yielding the same solution, i.e., $\UEposEst^{\MMLfast} = \UEposEst^{\MMLa}$.


\subsection{Simulation Framework and Metrics}

To evaluate the localization performance improvement of the proposed method against the baselines, we report the \gls{rmse} of the position estimates
\begin{equation}
    \rmse\{\UEposEst^{\mathrm{i}}\} = \sqrt{\Exp{\norm{\UEposEst^{\mathrm{i}} - \UEpos}^2}},
\end{equation}
where $\mathrm{i}$ refers to the chosen estimator.
We also use the Hit-Rate, defined as the proportion of estimates falling within the main lobe of the likelihood, i.e., where the absolute position error is smaller than half the intrinsic range resolution \cite{richards_principles_2010}
\begin{equation}
    \rangeRes \triangleq \frac{c}{\BW} = \frac{\carrierF}{\numF \Fspacing} \; \carrierWl.
\end{equation}
Both metrics are given as a function of the per-node average \gls{snr}, defined as
\begin{equation}
    \snr = \frac{\E_{\RXindex,\channelcoeff} \left\{ \abs{\channelcoeff[\RXindex]}^2 \right\} \constVar}{\noiseVar},
\end{equation}
where $\constVar$ denotes the symbol variance, assumed identical for pilot and data symbols and normalized to $\constVar=1$.
Under the stationarity and \gls{los} assumptions, no fading is modeled.
Following the physical optics approximation \cite{kishk_high_2011}, the channel coefficients are simulated as
\begin{equation}
    \channel[\RXindex] = e^{\jc \channelPhase_{\RXindex}} \frac{1}{\norm{\UEpos - \RXpos{\RXindex}}},
\end{equation}
where $\channelPhase_{\RXindex} \sim \U_{[0,2\pi)}$ accounts for a random phase at each node $\RXindex$.

While the developed framework is general and applies to any \gls{das} geometry, the simulations consider a \glsdesc{uca} of radius $\SRXradius$ and aperture $\SRXaperture \in (0, 2\pi]$.
The \gls{ue} position is drawn uniformly at random from a disk of radius\footnotemark $\Sradius < \SRXradius$.
\footnotetext{A small margin is introduced to avoid the singularity of the $\norm{\UEpos - \RXpos{\RXindex}}^{-1}$ path-loss model as $\norm{\UEpos - \RXpos{\RXindex}}\rightarrow0$.}
The attenuation at each node differs, while $\noiseVar$ is held constant across nodes and determined from the average \gls{snr}.
The latter is defined as $\snr = 2 / (\Sradius^2 \noiseVar)$, which follows from the approximation \begingroup\small$\E_{\RXindex, \channelcoeff}\{ \abs{\channelcoeff[\RXindex]}^2\} = \E_{\RXindex, \UEpos}\{ \norm{\UEpos - \RXpos{\RXindex}}^{-2} \} \approx 2/\Sradius^2$\endgroup, for this configuration.

Unless otherwise stated, the simulation parameters are listed in \autoref{tab:parameters}, where the selected frequency band follows prospective recommendations for \gls{isac} in \gls{6g} systems \cite{baduge_frequency_2025,bazzi_coverage_2026}.

\begin{table}[t]
    \centering
    \caption{Default Parameters used in Simulation.}
    \label{tab:parameters}
        \begin{tabular}{|l|c|}
            \hline
            \multicolumn{2}{|c|}{\textbf{\gls{ue}}} \\ \hline \hline
            Carrier frequency $\carrierF$ & $\SI{7.2}{\giga\hertz}$ \\
            Subcarrier spacing $\Fspacing$ & $\SI{45}{\kilo\hertz}$ \\ 
            Number of subcarriers $\numF$ & $160$ \\
            Number of \gls{ofdm} pilot symbols $\numP$ & $1$ \\
            Number of \gls{ofdm} data symbols $\numD$ & $35$ \\
            Pilots & \glsentryshort{bpsk} \\
            Data constellation $\constSet_{\constMap(\Findex,\Dindex)}$ & $256$-\gls{qam} ($\forall \Findex,\Dindex$)\\ 
            \hline \hline 
            \multicolumn{2}{|c|}{\textbf{\gls{srx}}} \\ \hline \hline 
            Number of nodes $\numRX$ & $8$ \\
            Radius $\SRXradius$ & $\SI{5000}{\carrierWl} = \SI{208.3}{\meter}$ \\
            Aperture $\SRXaperture$ & $2\pi~\si{\radian}$ \\
            \hline \hline 
            \multicolumn{2}{|c|}{\textbf{Simulation}} \\ \hline \hline 
            Scene radius $\Sradius$ & $\SI{4800}{\carrierWl} = \SI{200}{\meter}$ \\
            \makecell[l]{Number of grid points $\Ngrid$ \\ \begingroup\scriptsize(square area of side $2\Sradius$)\endgroup} & $40^2=1600$ \\
            Refinement optimization method & $\texttt{Nelder-Mead}$ \cite{nelder_simplex_1965} \\
            Number of \gls{mc} iterations $\Nmc$ & $3000$ \\
            \hline 
        \end{tabular}
\end{table}

\subsection{Localization Performance}

First, we characterize the performance of the proposed approximations with respect to $\JMLc$, and then compare them against the baselines.

\subsubsection{Quality of the Approximations}
\autoref{fig:RMSE_JML_comparison} depicts the \gls{rmse} as a function of the \gls{snr} for $\numD=16$ and $\numD=36$, respectively, and with a single \gls{ofdm} pilot symbol ($\numP=1$).
A small number of subcarriers is chosen ($\numF=40$) to prevent the computational runtime of $\JMLc$ from becoming prohibitive, as discussed in Section~\ref{sec:Complexity}.
On one hand, $\JMLc$ outperforms $\JMLa$, as expected, and the performance gap introduced by the approximation increases with $\numD$.
As discussed in Section~\ref{sec:Complexity}, this comes at a significant increase in computational requirements (see \autoref{fig:Time_vs_D_JML_comparison}).
On the other hand, the low-rank acceleration $\JMLfast$ performs similarly to $\JMLa$, though marginally worse at very low \gls{snr}.

\begin{figure}[ht]
    \includegraphics[width=0.95\linewidth]{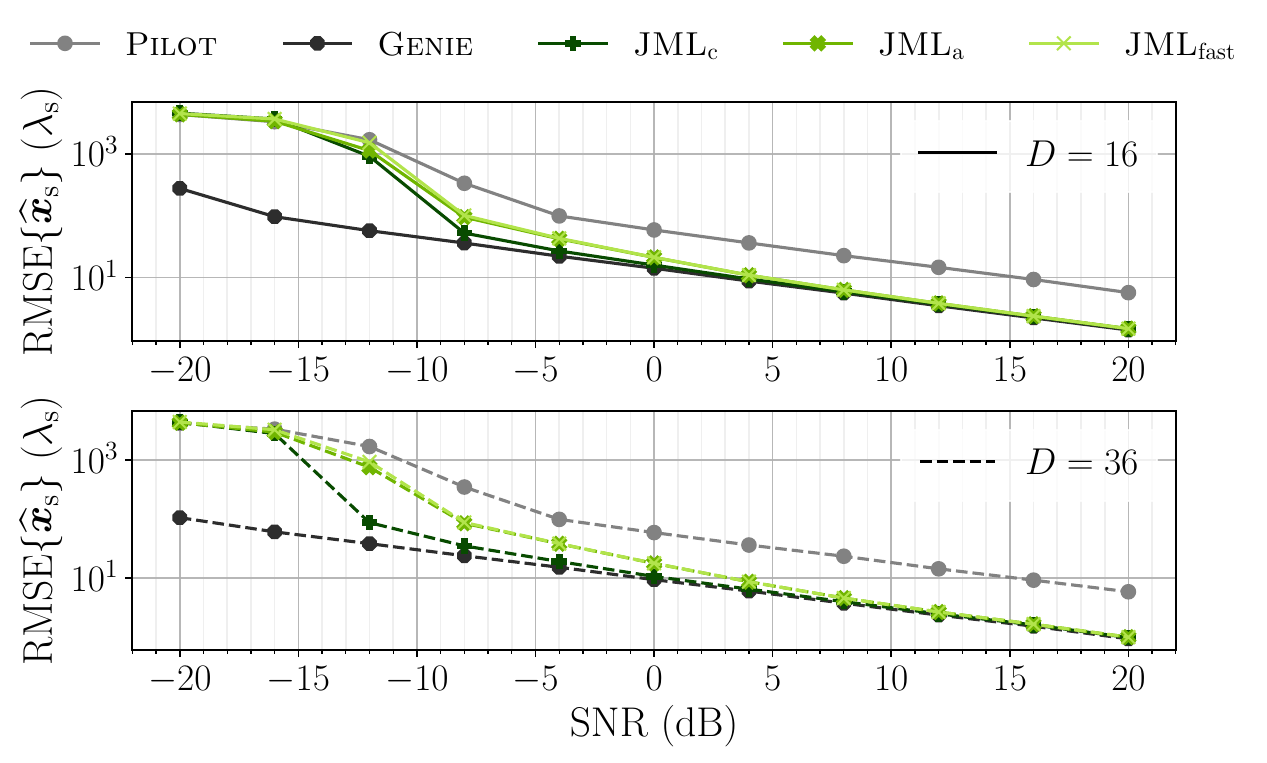}
    \caption{\gls{rmse} as a function of the \gls{snr} for different number of data symbols $\numD$. Parameters: $\numF = 40$, $\Fspacing = \SI{180}{\kilo\hertz}$ (yielding the same bandwidth as in \autoref{tab:parameters}); all other parameters are left at their default values.
    Note that the \gls{rmse} is normalized with respect to $\carrierWl$ (here \SI{4.17}{\centi\meter}).}
    \label{fig:RMSE_JML_comparison}
\end{figure}

\subsubsection{Comparison with Baselines}
\autoref{fig:RMSE_SNR_constellations} shows the \gls{rmse} and Hit Rate of the proposed estimators $\JMLa$ and $\JMLfast$ against the considered baselines, as a function of the \gls{snr} and for different data constellations.
The following observations can be drawn:
\begin{itemize}
    \item As established in \cite{reniers_localization_2026}, the \textit{distributed} \gls{dd} baseline underperforms its \textit{centralized} counterpart, while both converge to the $\PDmethod$ bound when all data symbols are correctly demodulated.
    \item The proposed estimators $\JMLa$, $\JMLfast$, and $\MMLfast$ all achieve superior localization performance over the pilot-only baseline, with an \gls{rmse} reduction of up to a factor of $5.7$. 
    Furthermore, they converge to the $\PDmethod$ performance at a lower \gls{snr} than the \gls{dd} baselines, with an \gls{snr} gain of up to $\SI{4.8}{\decibel}$ over the latter.
    \item The low-rank acceleration does not degrade the performance of $\JMLa$ over a large \gls{snr} range, thereby providing a computational gain with no performance penalty. 
    As highlighted by the Hit Rate curves, a marginal degradation in localization capability is observed at very low \gls{snr}, as anticipated in Section~\ref{sec:JML_fast}.
\end{itemize}

\begin{figure*}
    \includegraphics[width=1.0\linewidth]{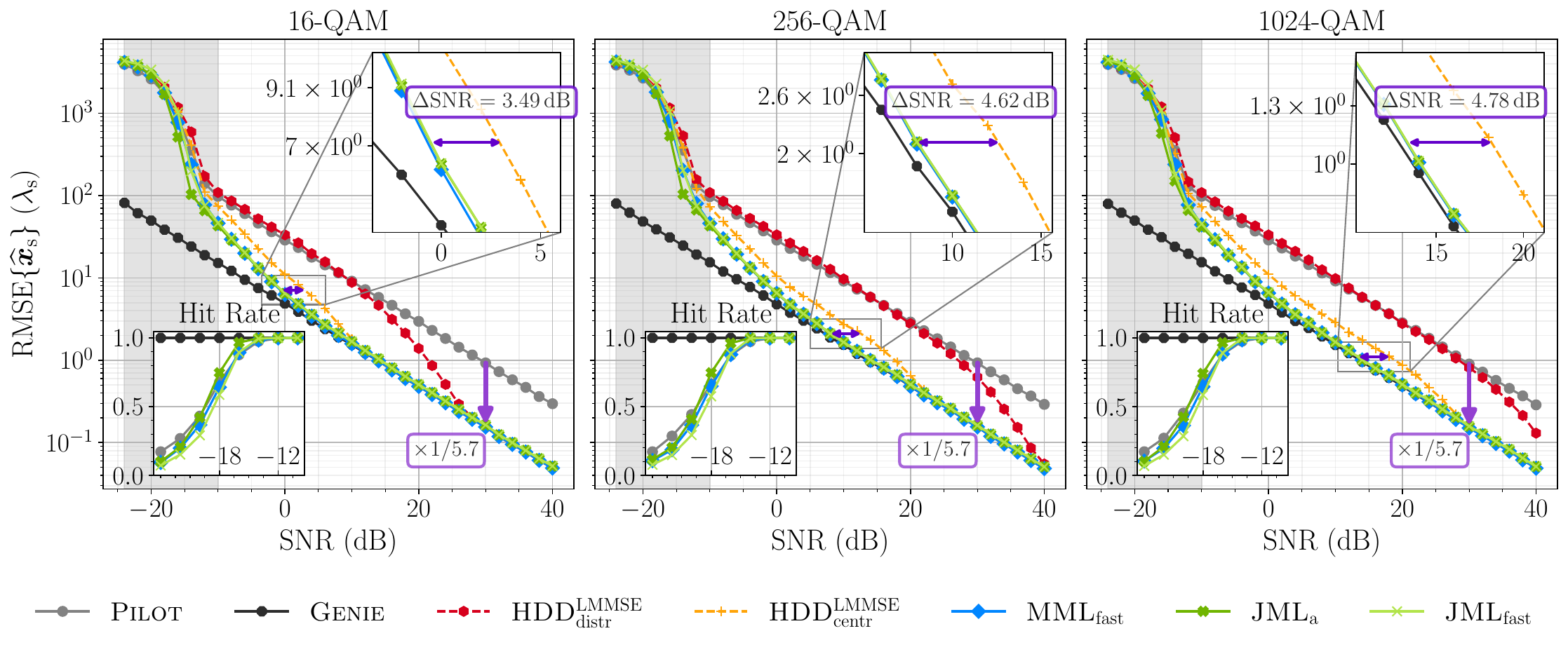}
    \caption{\gls{rmse} and Hit-Rate as a function of the \gls{snr} (step of \SI{2}{\decibel}) for different data constellations. The Hit-Rate is only displayed over the \gls{snr} range $[-24, 10]~\si{\decibel}$ (indicated by the gray area), as all methods achieve $\SI{100}{\percent}$ Hit-Rate at higher \gls{snr} values. The \gls{rmse} is normalized by $\carrierWl$ (here \SI{4.17}{\centi\meter}).
    }
    \label{fig:RMSE_SNR_constellations}
\end{figure*}

While $\MMLfast$ and the proposed \gls{jml} estimators exhibit similar localization performance, both unaffected by increasing modulation order, the key advantage of $\JMLfast$ (and $\JMLa$) lies in the computational cost required to achieve this performance.
As discussed in Section~\ref{sec:Complexity} (see \autoref{fig:Time_vs_ConstSize}), the complexity of $\MMLfast$ grows with constellation size, while that of $\JMLa$ and $\JMLfast$ remains invariant and considerably lower.
Furthermore, the proposed \gls{jml} estimators are fully \textit{constellation-agnostic}: regardless of the constellation employed, the \gls{cpu} requires no knowledge of it, and both localization performance and computational cost remain unchanged.
The \gls{mml} approach, on the other hand, requires knowledge of the constellation—which is not guaranteed in practical opportunistic scenarios—as well as prior distribution assumptions on the \glspl{np}.
Finally, it is worth noting that $\JMLa$ marginally outperforms $\MMLfast$ at low \gls{snr}, as reflected in both the \gls{rmse} and Hit Rate curves.

\vspace*{0.5cm}
\subsubsection{Impact of System Parameters}

This section analyzes the impact of system parameters on localization performance.
A summary of the results, considering both the performance variations discussed here and the computational cost analyzed in Section~\ref{sec:Complexity}, is provided in \autoref{tab:performance_analysis}.\\[0.2em]

\begin{table}[ht]
    \centering
    \setlength{\tabcolsep}{3pt}
    \renewcommand{\arraystretch}{1.2}
    \caption{Impact of system parameters on localization \gls{rmse} (R) and computational costs (C). Note: $\Downarrow$ indicates a stronger reduction than $\downarrow$, and \textcolor{gray}{$\uparrow_{\mathrm{n}}$} indicates negligible increase.}
    \label{tab:performance_analysis}

    \resizebox{\linewidth}{!}{
    \begin{tabular}{|l||cc|cc|cc|cc|cc||cc|}
        \hline

        & \multicolumn{2}{c|}{\textcolor{Pcolor}{$\Pmethod$}}
        & \multicolumn{2}{c|}{\textcolor{DDcentrcolor}{$\mathrm{HDD}$}}
        & \multicolumn{2}{c|}{\textcolor{MMLfastcolor}{$\MMLfast$}}
        & \multicolumn{2}{c|}{\textcolor{JMLacolor}{$\JMLa$}}
        & \multicolumn{2}{c||}{\textcolor{JMLfastcolor}{$\JMLfast$}}
        & \multicolumn{2}{c|}{Figure} \\

        \cline{2-13}

        & R & C
        & R & C
        & R & C
        & R & C
        & R & C
        & R & C \\

        \hline \hline

        $\uparrow \numP$
        & $\downarrow$ & \textcolor{gray}{$\uparrow_{\mathrm{n}}$}
        & $\downarrow$ & \textcolor{gray}{$\uparrow_{\mathrm{n}}$}
        & $\downarrow$ & \textcolor{gray}{$\uparrow_{\mathrm{n}}$}
        & $\downarrow$ & \textcolor{gray}{$\uparrow_{\mathrm{n}}$}
        & $\downarrow$ & \textcolor{gray}{$\uparrow_{\mathrm{n}}$}
        & -- & -- \\ \hline

        $\uparrow \numD$
        & -- & --
        & $\downarrow$ & $\uparrow$
        & $\downarrow$ & $\uparrow$
        & $\downarrow$ & $\uparrow$
        & $\downarrow$ & \textcolor{gray}{$\uparrow_{\mathrm{n}}$}
        & \ref{fig:RMSE_JML_comparison} & \ref{fig:Time_vs_D_JML_comparison} \\ \hline
        
        $\uparrow \constSize$
        & -- & --
        & \colorbox{gray!20}{$\uparrow$} & \colorbox{gray!20}{$\uparrow$}
        & \colorbox{gray!20}{--} & \colorbox{gray!20}{$\Uparrow$}
        & \colorbox{gray!20}{--}  & \colorbox{gray!20}{--} 
        & \colorbox{gray!20}{--}  & \colorbox{gray!20}{--} 
        & \ref{fig:RMSE_SNR_constellations} & \ref{fig:Time_vs_ConstSize} \\ \hline

        $\uparrow \numRX$
        & $\Downarrow$ & $\uparrow$
        & $\Downarrow$ & $\uparrow$
        & $\Downarrow$ & $\uparrow$
        & $\Downarrow$ & $\uparrow$
        & $\Downarrow$ & $\uparrow$
        & \ref{fig:RMSE_vs_N} & -- \\ \hline

        $\uparrow \numF$
        & $\Downarrow$ & $\uparrow$
        & $\Downarrow$ & $\uparrow$
        & $\Downarrow$ & $\uparrow$
        & $\Downarrow$ & $\uparrow$
        & $\Downarrow$ & $\uparrow$
        & \ref{fig:RMSE_vs_Q} & -- \\ \hline

    \end{tabular}
    }
\end{table}

\autoref{fig:RMSE_vs_N} shows the \gls{rmse} as a function of $\numRX$ for multiple \gls{snr} values.
Increasing $\numRX$ naturally reduces the \gls{rmse} for all methods, though the performance gap between methods depends on the operating \gls{snr}, i.e., whether the observations carry near-perfect positioning information or not.
Nevertheless, at very low \gls{snr} (e.g., \SI{-16}{\decibel}), increasing $\numRX$ favors the proposed \gls{jml} estimators, while adding nodes at this noise level does not yield consistent improvement for the \gls{mml} method.\\[0.2em]

\autoref{fig:RMSE_vs_Q} depicts the \gls{rmse} with respect to $\numF$ for different \gls{snr} values.
As expected, increasing $\numF$ improves localization performance across all methods.
Similarly, the performance gap between the methods varies with the \gls{snr}.
Below the \gls{snr} threshold at which all methods converge to the $\PDmethod$ bound, the \gls{jml} and \gls{mml} estimators benefit more substantially from additional subcarriers than the \gls{dd} baselines.\\[0.2em]

\begin{figure}[!ht]
    \includegraphics[width=0.95\linewidth]{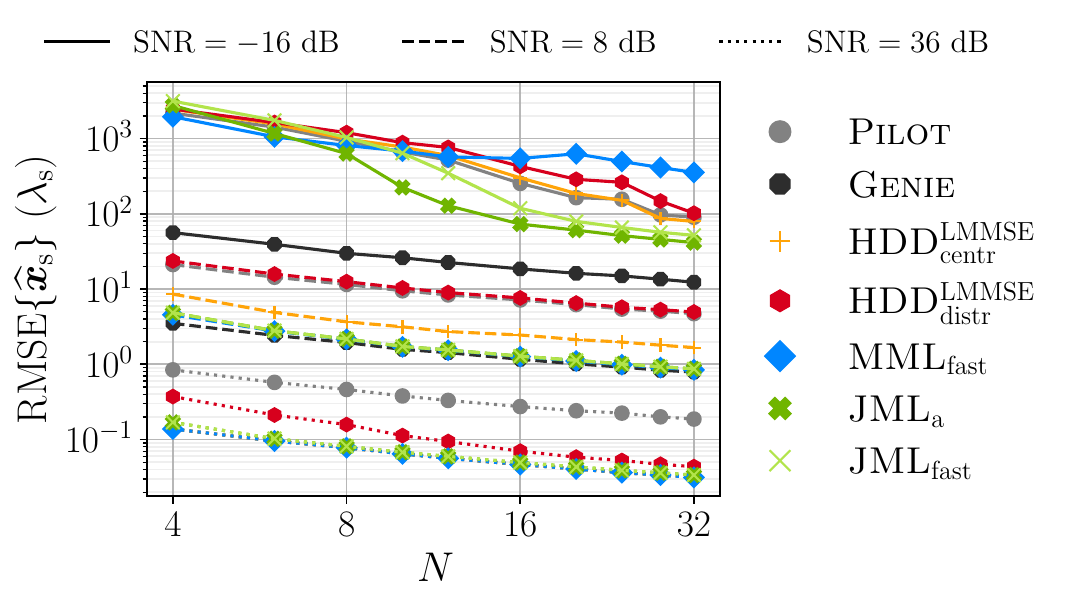}
    \caption{\gls{rmse} as a function of the number of nodes $\numRX$, for different \gls{snr} values. All other parameters are left at their default values.
    The \gls{rmse} is normalized with respect to $\carrierWl$ (here \SI{4.17}{\centi\meter}).
    Note that $\JMLa$, $\JMLfast$, $\HDDcentrmethod$, and $\MMLfast$ all coincide with the $\PDmethod$ bound at $\snr = \SI{36}{\decibel}$.}
    \label{fig:RMSE_vs_N}
\end{figure}

\begin{figure}[!ht]
    \includegraphics[width=0.95\linewidth]{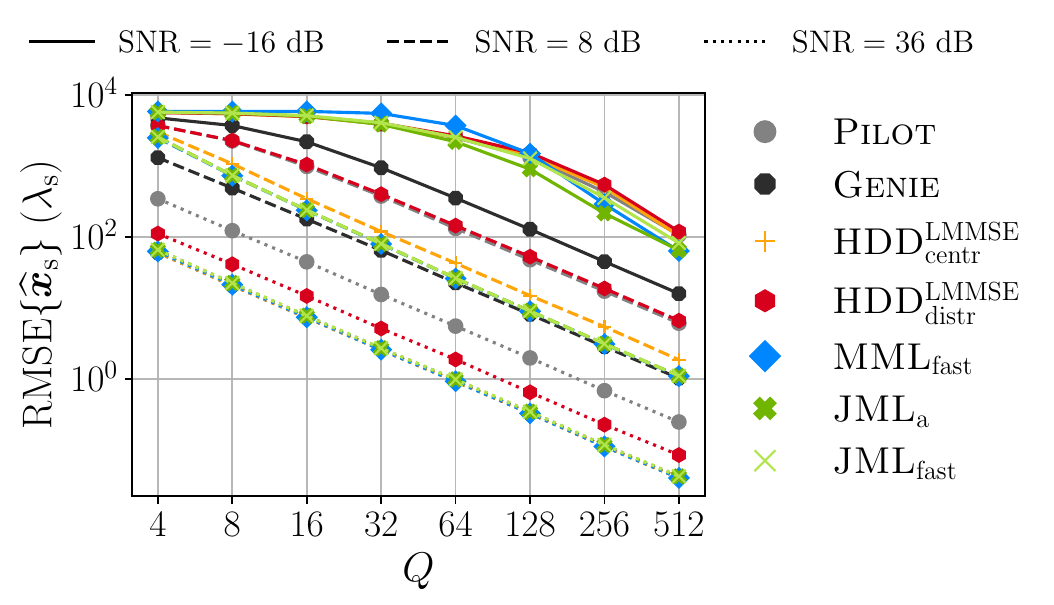}
    \caption{\gls{rmse} as a function of the number of nodes $\numF$, for different \gls{snr} values. All other parameters are left at their default values.
    The \gls{rmse} is normalized with respect to $\carrierWl$ (here \SI{4.17}{\centi\meter}).
    Note that $\JMLa$, $\JMLfast$, $\HDDcentrmethod$, and $\MMLfast$ all coincide with the $\PDmethod$ bound at $\snr = \SI{36}{\decibel}$.
    }
    \label{fig:RMSE_vs_Q}
\end{figure}

\subsection{Hybrid ToA-TDoA Interpretation}\label{sec:JML_ToA_TDoA}

This section provides an interpretation of the proposed \gls{jml} estimators that offers insight into the 
geometric behavior of the proposed methods.
Specifically, we show that the proposed \gls{jml} estimators operate as \textit{single-step hybrid \gls{toa}--\gls{tdoa}} estimators.
This is derived analytically for $\JMLa$, which, as shown previously, is well approximated by $\JMLfast$. 
Additionally, a similar behavior is observed for $\JMLc$.
For convenience, we define the pilot and data terms of $\JMLa$ \eqref{eq:JMLa_estimator} as $\pilotTermLL(\UEposTest)$ and $\dataTermLL(\UEposTest)$, respectively.

Using the steering matrix definition into the pilot term of \eqref{eq:JMLa_estimator}, the latter becomes
\begingroup\small
\begin{equation}
    \pilotTermLL(\UEposTest) = \frac{1}{\pilotEnergy} \sum_{\RXindex=0}^{\numRX-1} \abs{\sum_{\Findex=0}^{\numF-1} \pilotObsEq^{*}[\RXindex,\Findex] e^{-\jc 2 \pi \ToA(\UEposTest, \RXpos{\RXindex}) \Findex\Fspacing}}^2,
\end{equation}
\endgroup
which reveals the \gls{toa} between the candidate position and the $\RXindex$-th node, $\ToA(\UEposTest, \RXpos{\RXindex}) \triangleq \norm{\UEposTest - \RXpos{\RXindex}} /c$, where $c = \carrierWl \carrierF$.
Substituting the observation model in the noise-free case yields the following ambiguity function:
\begingroup\small
\begin{align}
     & \pilotTermLL^{\mathrm{nf}}(\UEposTest) = \frac{1}{\pilotEnergy} \sum_{\RXindex=0}^{\numRX-1}  \abs{\channelcoeff[\RXindex]}^2 \nonumber \\
     & \times  \Bigg| \sum_{\Findex=0}^{\numF-1} \sum_{\Pindex=0}^{\numP-1} \abs{\pilotMat[\Findex,\Pindex]}^2  
     e^{-\jc 2 \pi \left[ \ToA(\UEposTest, \RXpos{\RXindex}) - \ToA(\UEpos, \RXpos{\RXindex}) \right] \Findex\Fspacing}  \Bigg| ^2. \label{eq:JMLa_NF_pilot}
\end{align}
\endgroup
For a given node $\RXindex$, $\pilotTermLL^{\mathrm{nf}}(\UEposTest)$ is maximized when $\ToA(\UEposTest, \RXpos{\RXindex}) = \ToA(\UEpos, \RXpos{\RXindex})$, i.e., when the \gls{toa}, or equivalently the range $\norm{\UEposTest - \RXpos{\RXindex}}$, is correctly estimated.
Combining information across nodes defines a \textit{set of range loci} whose intersection yields the position estimate.

Applying the same procedure to the data term yields
\begingroup\small
\begin{align}
    & \dataTermLL(\UEposTest) = \frac{1}{\channelConstructEnergy(\UEposTest)} \sum_{\Findex=0}^{\numF-1} \sum_{\Dindex=0}^{\numD-1}   \Bigg| \sum_{\RXindex=0}^{\numRX-1} \dataObs^{*}[\RXindex,\Findex,\Dindex] \nonumber \\ 
    & \hspace{3cm}  \times \channelcoeffEstPilots(\UEposTest)[\RXindex]  e^{-\jc 2 \pi \ToA(\UEposTest, \RXpos{\RXindex}) \Findex\Fspacing} \Bigg| ^2, \\
    & =  \frac{1}{\channelConstructEnergy(\UEposTest)} \sum_{\Findex=0}^{\numF-1} \sum_{\Dindex=0}^{\numD-1}  \sum_{\RXindex,\RXindex'=0}^{\numRX-1}   \dataObs^{*}[\RXindex,\Findex,\Dindex] \dataObs[\RXindex',\Findex,\Dindex] \nonumber \\
    & \hspace{1cm}  \times  \channelcoeffEstPilots(\UEposTest)[\RXindex] \channelcoeffEstPilots^{*}(\UEposTest)[\RXindex']  e^{-\jc 2 \pi \TDoA(\UEposTest,\RXpos{\RXindex},\RXpos{\RXindex'}) \Findex\Fspacing}, 
\end{align}
\endgroup
which reveals, upon expanding the squared modulus into a double sum, the \gls{tdoa} between nodes $\RXindex$ and $\RXindex'$: $\TDoA(\UEposTest,\RXpos{\RXindex},\RXpos{\RXindex'}) \triangleq \ToA(\UEposTest, \RXpos{\RXindex}) - \ToA(\UEposTest, \RXpos{\RXindex'})$.
Substituting the observation model in the noise-free case results in the following ambiguity function:
\begingroup\small
\begin{align}
    & \dataTermLL^{\mathrm{nf}}(\UEposTest) = \frac{1}{\channelConstructEnergy(\UEposTest)} \sum_{\Findex=0}^{\numF-1} \sum_{\Dindex=0}^{\numD-1} \abs{\dataMat[\Findex,\Dindex]}^2 \nonumber \\
    & \times \abs{\sum_{\RXindex=0}^{\numRX-1}  \channelcoeffEstPilots^{*}(\UEposTest)[\RXindex]\channelcoeff[\RXindex] e^{-\jc 2 \pi \left[ \ToA(\UEposTest, \RXpos{\RXindex}) - \ToA(\UEpos, \RXpos{\RXindex}) \right] \Findex\Fspacing}}^2 \label{eq:JMLa_NF_data} \\
    & = \frac{1}{\channelConstructEnergy(\UEposTest)} \sum_{\Findex=0}^{\numF-1} \sum_{\Dindex=0}^{\numD-1} \abs{\dataMat[\Findex,\Dindex]}^2  \sum_{\RXindex,\RXindex'=0}^{\numRX-1} \channelcoeffEstPilots^{*}(\UEposTest)[\RXindex] \channelcoeffEstPilots(\UEposTest)[\RXindex'] \nonumber  \\
    & \times  \channelcoeff[\RXindex] \channelcoeff^{*}[\RXindex'] e^{-\jc 2 \pi \left[ \TDoA(\UEposTest,\RXpos{\RXindex},\RXpos{\RXindex'}) - \TDoA(\UEpos,\RXpos{\RXindex},\RXpos{\RXindex'})  \right] \Findex \Fspacing} ,
\end{align}
\endgroup
which is maximized when $\TDoA(\UEposTest,\RXpos{\RXindex},\RXpos{\RXindex'}) = \TDoA(\UEpos,\RXpos{\RXindex},\RXpos{\RXindex'})$, i.e., when the \gls{tdoa} is correctly estimated.

While the pilot and data terms share a similar structure, they differ in the \gls{np} being eliminated and consequently in the coherence structure. 
Notably, the pilot term involves perfect knowledge of the pilot symbols, yielding $\abs{\pilotMat[\Findex,\Pindex]}^2$ in \eqref{eq:JMLa_NF_pilot}, whereas the data term in \eqref{eq:JMLa_NF_data} involves the product $\channelcoeffEstPilots^{*}(\UEposTest)[\RXindex]\channelcoeff[\RXindex] \approx \abs{\channelcoeff[\RXindex]}^2$, which holds when the candidate position $\UEposTest$ is close to the true position $\UEpos$ and the \gls{snr} is sufficiently high.
Consequently, the \gls{tdoa} interpretation arises from the fact that coherence is maintained across the spatial dimension: $\dataTermLL(\UEposTest)$ accumulates \gls{tdoa} contributions over time-frequency resource elements, while $\pilotTermLL(\UEposTest)$ accumulates \gls{toa} contributions across nodes.
Therefore, the data term defines a \textit{set of hyperbolic loci} \cite{torrieri_statistical_1984} whose intersection yields the position estimate.
Importantly, despite this hybrid \gls{toa}-\gls{tdoa} interpretation, the proposed estimators remain single-step, as these quantities are implicitly captured by the objective function rather than being pre-estimated. 
These observations are illustrated in \autoref{fig:LL_loci}, which displays the objective function of the pilot term and the data term for our methods.
The top left panel shows $\pilotTermLL(\UEposTest)$, where the \gls{toa} loci intersect at the true \gls{ue} position.
The bottom row shows the data term of $\JMLc$\footnotemark, $\JMLa$, and $\JMLfast$, which exhibit hyperbolic loci intersecting at $\UEpos$, confirming the hybrid \gls{toa}-\gls{tdoa} interpretation.
While $\JMLfast$ closely matches $\JMLa$ at this \gls{snr}, a notable discrepancy with $\JMLc$ is observed, which can be attributed to the \gls{np} separation approximation \eqref{eq:approximation} introduced to ensure tractability, whereas $\JMLc$ handles both \glspl{np} jointly and optimally.
Finally, the data term of $\MMLfast$ exhibits a similar structure to that of the pilot term, where the broadening of the circular loci reflects the data uncertainty introduced by marginalizing over the constellation prior.

\footnotetext{The $\JMLc$ estimator does not explicitly incorporate a separate pilot term, as it is already embedded in the matrix $\ARQnumChange(\UEposTest)$. Furthermore, under typical parameter settings where $\numD \gg \numP$, the data contribution dominates the total objective function of all methods, hence the predominantly \gls{tdoa}-based behavior of $\JMLc$.}

\begin{figure}[ht]
    \includegraphics[width=\linewidth]{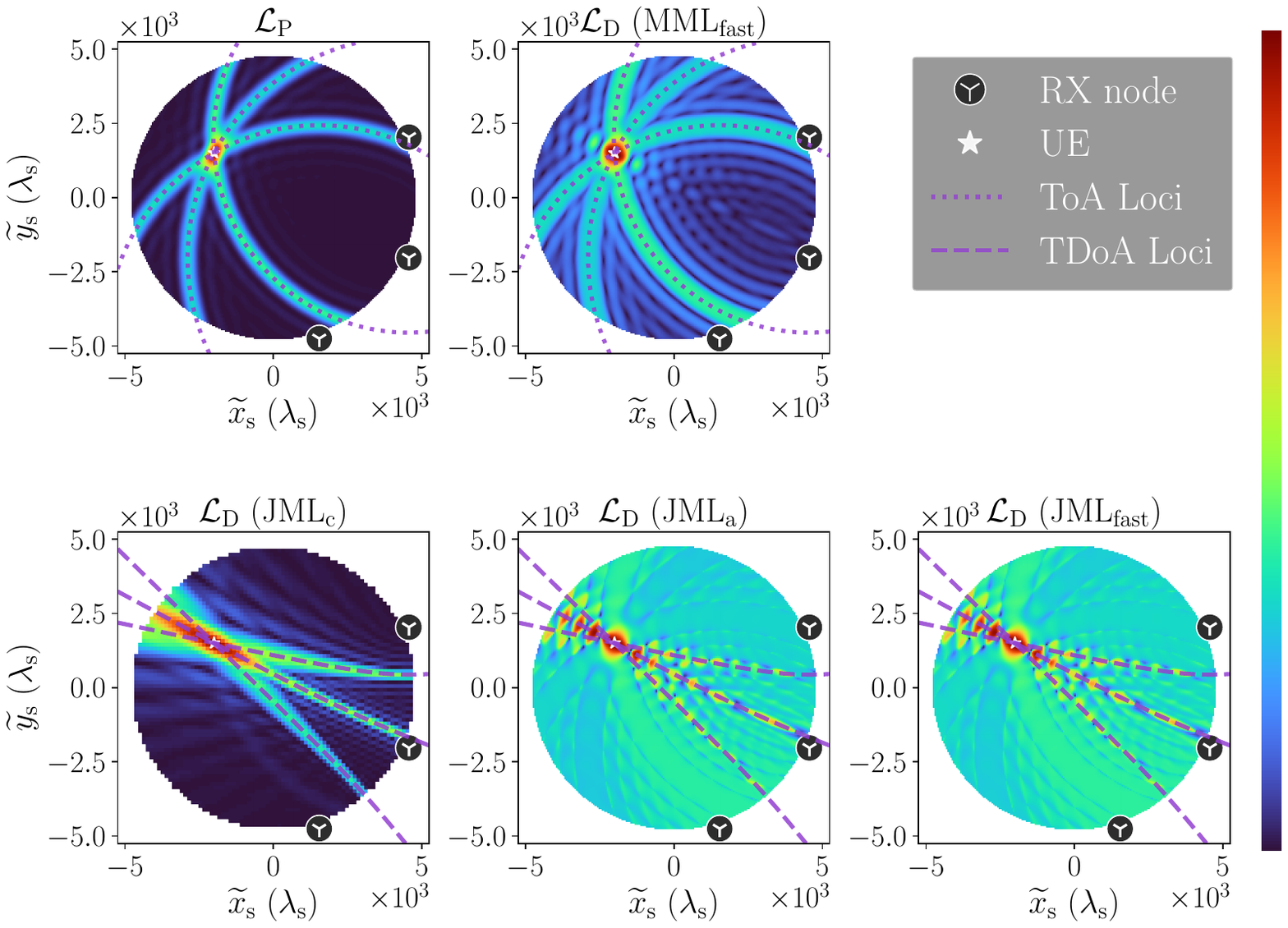}
    \caption{Illustration of the objective function terms, from left to right and top to bottom: pilot term, data term of $\MMLfast$, $\JMLc$, $\JMLa$, and $\JMLfast$. Parameters: $\numRX = 3$, $\Fspacing = \SI{90}{\kilo\hertz}$, $\SRXaperture = 0.8\pi~\si{\radian}$, $\snr = \SI{30}{\decibel}$, $\Ngrid = 200^2$; all other parameters are set to their default values. \gls{toa} loci are shown for each node $\RXindex$, corresponding to positions $\UEposTest$ satisfying $\norm{\UEposTest - \RXpos{\RXindex}} = \norm{\UEpos - \RXpos{\RXindex}}$. \gls{tdoa} loci are shown for each node pair $(\RXindex, \RXindex')$, corresponding to positions $\UEposTest$ satisfying $\norm{\UEposTest - \RXpos{\RXindex}} - \norm{\UEposTest - \RXpos{\RXindex'}} = \norm{\UEpos - \RXpos{\RXindex}} - \norm{\UEpos - \RXpos{\RXindex'}}$.}
    \label{fig:LL_loci}
\end{figure}

\section{Complexity Analysis}\label{sec:Complexity}

This section analyzes the computational complexity of the proposed method against the considered baselines, first in terms of asymptotic operation counts, and then through runtime measurements, illustrating the influence of system parameters on computational requirements.

\subsection{Theoretical Analysis}\label{sec:Complexity_th}
\autoref{tab:complexity} reports the asymptotic complexity of each processing step for an $\constSize$-ary constellation, i.e., $\forall \Findex,\Dindex \, \#\constSet_{\constMap(\Findex,\Dindex)} = \constSize$.
The localization cost includes only the objective function maximization over a grid of $\Ngrid$ candidate positions\footnotemark.
\footnotetext{After grid evaluation, estimates are refined via an optimization algorithm requiring additional objective function evaluations. Since the number of such evaluations depends on the solver strategy rather than on the proposed 
estimators, it is excluded from the complexity analysis. 
In practice, a sufficiently fine search grid yields adequate localization accuracy without refinement.}
The first four columns correspond to the baselines, while the remaining columns correspond to the proposed methods.
Note that the complexity of $\PDmethod$ follows directly from that of $\Pmethod$ by substituting $\numP$ with $\numP + \numD$.

The \textit{almost-optimal} estimator $\JMLc$ \eqref{eq:JMLc_estimator} incurs an asymptotic complexity of $\Compl(\Ngrid \numF^2\numD^2 \numRX) + \Compl(\Ngrid\numF^3\numD^3)$. 
The first term accounts for computing the matrix $\ARQnumChange(\UEposTest)$ at all candidate positions, while the 
second arises from computing its largest eigenvalue, assuming worst-case full eigendecomposition\footnotemark.
As previously mentioned, the resulting computational complexity is intractable, except for very small \glspl{tfb} (see \autoref{fig:Time_vs_D_JML_comparison}).
The \textit{approximate} estimator $\JMLa$ \eqref{eq:JMLa_estimator} incurs a complexity of $\Compl(\numRX\numF\numP)$ for constructing $\pilotObsEq$, and a localization cost of $\Compl(\Ngrid\numRX\numF)$ for the pilot term and $\Compl(\Ngrid\numRX\numF\numD)$ for the data term.
The \textit{low-rank acceleration} $\JMLfast$ \eqref{eq:JMLfast_estimator} leaves the pilot term unchanged, while reducing the data term localization cost to $\Compl(\Ngrid\numRX\numF)$, which merges asymptotically with the pilot term (hence with a factor of two in practice).
However, it introduces an additional cost of $\Compl(\numF \max(\numRX,\numD) \min(\numRX,\numD)^2)$ 
\cite{li_tutorial_2019} for computing the first \gls{svd} component of $\dataObsMat{\Findex}$ across all $\Findex$, which reduces to $\Compl(\numRX^2\numF\numD)$ in the typical case $\numD \geq \numRX$.
\newcounter{sharedfn}
\setcounter{sharedfn}{\value{footnote}}
Note that this again assumes a worst-case full \gls{svd} decomposition\footnotemark[\value{sharedfn}].
As shown in Section~\ref{sec:Complexity_meas}, this cost is considerably lower than that of the localization step, since $\numRX \numD \ll \Ngrid$ in practice.

\footnotetext[\value{sharedfn}]{More efficient algorithms exist \cite{sobczyk_deterministic_2025} but are outside the scope of this work. 
}

It can be directly concluded that the proposed estimators $\JMLa$ and $\JMLfast$ are always faster than the $\MMLfast$ baseline.
Compared to the \gls{dd} baselines, $\JMLa$ introduces an additional localization cost of $\Compl(\Ngrid\numF\numD\numRX)$, while the \gls{dd} baselines incur a demodulation cost of $\Compl(\constSize \numF\numD)$ and $\Compl(\constSize \numRX \numF\numD)$ for centralized and distributed demodulation, respectively.
Therefore, $\JMLa$ is faster than $\HDDcentrmethod$ when $\Ngrid\numRX \ll \constSize$, and faster than $\HDDdistrmethod$ when $\Ngrid \ll \constSize$.
The former condition is seldom met in practice, while the latter is satisfied for large constellations and restricted grid sizes.
Since the \gls{svd} computation cost is negligible, as discussed above, $\JMLfast$ is expected to be faster than both \gls{dd} baselines, while achieving similar localization performance to its non-accelerated counterpart over a large \gls{snr} range, as demonstrated in Section~\ref{sec:results}.
However, the constant factor of two omitted in the asymptotic cost of the localization step results in execution times comparable to those of the \gls{dd} baselines in practice, while still yielding superior localization performance.
These conclusions are validated by runtime measurements in the following section.

\begin{table*}[ht]
    \centering
    \caption{Asymptotic Computational Complexity per Processing Step}
    \label{tab:complexity}
    \resizebox{\textwidth}{!}{
    \begingroup 
    \setlength{\tabcolsep}{4pt} 
    \renewcommand{\arraystretch}{1.8}
    \begin{tabular}{|l||c|c|c|c|c|c|c|c|}
        \hline
        Step &  \textcolor{Pcolor}{\textbf{\Pmethod}} & \textcolor{DDcentrcolor}{$\bm{\HDDcentrmethod}$} & \textcolor{DDdistrcolor}{$\bm{\HDDdistrmethod}$} & \textcolor{MMLfastcolor}{$\bm{\MMLfast}$} \cite{reniers_localization_2026} & \textcolor{JMLccolor}{$\bm{\JMLc}$} \eqref{eq:JMLc_estimator} & \textcolor{JMLacolor}{$\bm{\JMLa}$} \eqref{eq:JMLa_estimator} & \textcolor{JMLfastcolor}{$\bm{\JMLfast}$} \eqref{eq:JMLfast_estimator} \\
        \hline \hline
        \makecell[l]{Channel estimation}    & — & $\Compl(\numRX\numF\numP)$ & $\Compl(\numRX\numF\numP)$ & — & — & — & — \\ 
        \makecell[l]{Soft data estimation} & — & $\Compl(\numRX\numF\numD)$ & $\Compl(\numRX\numF\numD)$ & — & — & — & — \\ 
        \makecell[l]{Hard data decision}  & — & $\Compl(\constSize\numF\numD)$ & $\Compl(\numRX\constSize\numF\numD)$ & — & — & — & — \\
        \hline
        \makecell[l]{Symbol equalization \\ ($\pilotObsEq$, $\ObsEq$, or $\ObsEqDD$ construction)} & $\Compl(\numRX\numF\numP)$ & $\Compl(\numRX\numF(\numP+\numD))$  & $\Compl(\numRX\numF(\numP+\numD))$ & $\Compl(\numRX\numF\numP)$ & — & $\Compl(\numRX\numF\numP)$ & $\Compl(\numRX\numF\numP)$  \\
        \hline
        \makecell[l]{Computation of first \gls{svd} \\ component of $\dataObsMat{\Findex}$ for all $\Findex$}  & —  & —  & — & —  & —  & — & $\Compl(\numRX^2 \numF \numD)$ \\
        \hline
        \makecell[l]{\raisebox{0pt}[4.2ex][3.2ex]{Localization}}  & $\Compl(\Ngrid \numRX\numF)$ & $\Compl(\Ngrid \numRX\numF)$  & $\Compl(\Ngrid \numRX\numF)$ & \makecell[c]{$\Compl(\Ngrid \numRX\numF)$ \\ $+ \Compl(\Ngrid \numF\numD (\numRX + \sqrt{\constSize}))$} & \makecell[c]{$\Compl(\Ngrid \numF^2\numD^2 \numRX)$ \\ $+ \Compl(\Ngrid\numF^3\numD^3)$} & \makecell[c]{$\Compl(\Ngrid \numRX\numF)$ \\ $+ \Compl(\Ngrid \numF\numD \numRX)$} & $\Compl(\Ngrid\numRX\numF)$  \\
        \hline
    \end{tabular}
    \endgroup
    }
\end{table*}

\vspace*{-0.35cm}
\subsection{Runtime Measurements}\label{sec:Complexity_meas}
\vspace*{-0.1cm}
While the asymptotic complexity analysis provides useful insight into how system parameters influence computational requirements, it should be interpreted as a trend indicator only.
To further characterize the computational cost of the considered methods and validate that the trends reported in Section~\ref{sec:Complexity_th} translate into effective computational gains in practice, average runtime measurements are reported for different parameter settings.
For each method, the reported time includes all processing steps, including grid search plus refinement for localization cost.

\subsubsection{Gain of the Accelerations}

\autoref{fig:Time_vs_D_JML_comparison} depicts the average computation time per position estimate for $\JMLc$, $\JMLa$, $\JMLfast$, and the pilot-only and $\PDmethod$ baselines, with $\pm$ one standard deviation shown as the shaded area.
Note that a smaller \gls{tfb} than the default is used to prevent the runtime of $\JMLc$ from becoming prohibitive.
The approximation $\JMLa$ yields orders-of-magnitude reductions in computational time over $\JMLc$ as $\numD$ increases, e.g., from $\SI{70}{\second}$ to $\SI{0.2}{\second}$ (a factor of \underline{$350$}) for $\numD=36$.
Furthermore, while the localization cost of $\JMLa$ still increases with $\numD$, that of $\JMLfast$ remains independent of $\numD$, further reducing the computational cost as anticipated in \autoref{tab:complexity}.
However, for this small \gls{tfb}, the gain is limited (e.g., $\SI{0.15}{\second}$ for $\JMLfast$ with $\numD=36$) and becomes more significant under the default parameter settings, as shown in the following.
Finally, consistent with the analysis in Section~\ref{sec:Complexity_th}, the cost of computing the required \gls{svd} component for $\JMLfast$ is negligible relative to the localization cost.
Additionally, the symbol equalization cost is lower than $\PDmethod$ for the proposed methods, though it remains negligible compared to the other processing steps.

\begin{figure}[!ht]
    \includegraphics[width=1.0\linewidth]{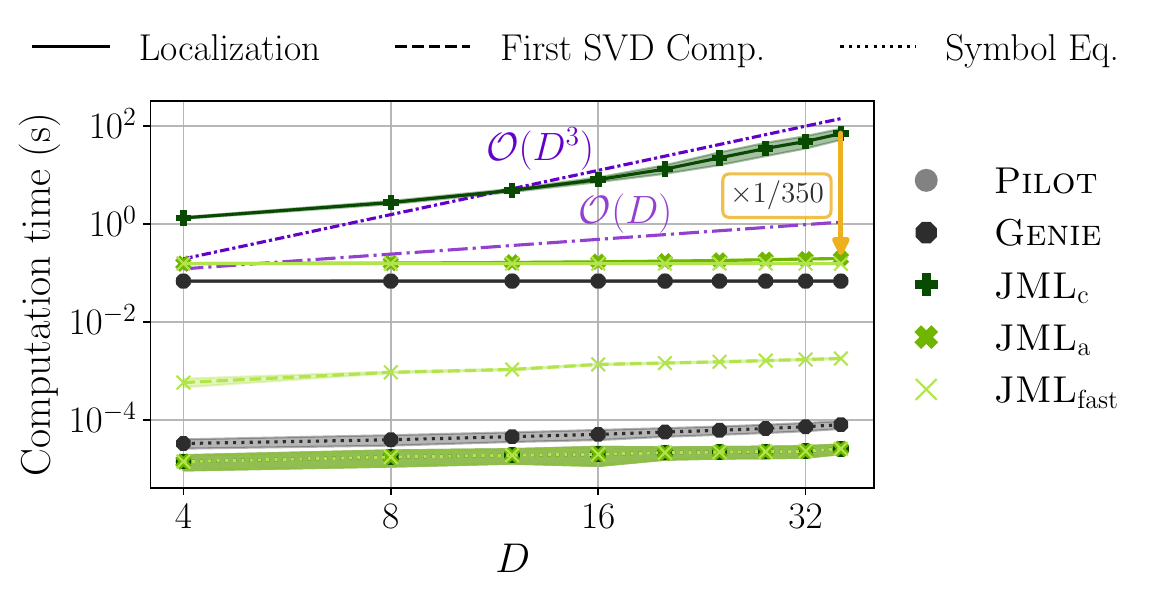}
    \caption{Mean computation time, with $\pm$ one standard deviation, as a function of $\numD$.  Parameters: $\numF = 40$, $\Fspacing = \SI{180}{\kilo\hertz}$ (yielding the same bandwith as in \autoref{tab:parameters}); all other parameters are left at their default values.}
    \label{fig:Time_vs_D_JML_comparison}
\end{figure}
\begin{figure}[ht]
    \includegraphics[width=1.0\linewidth]{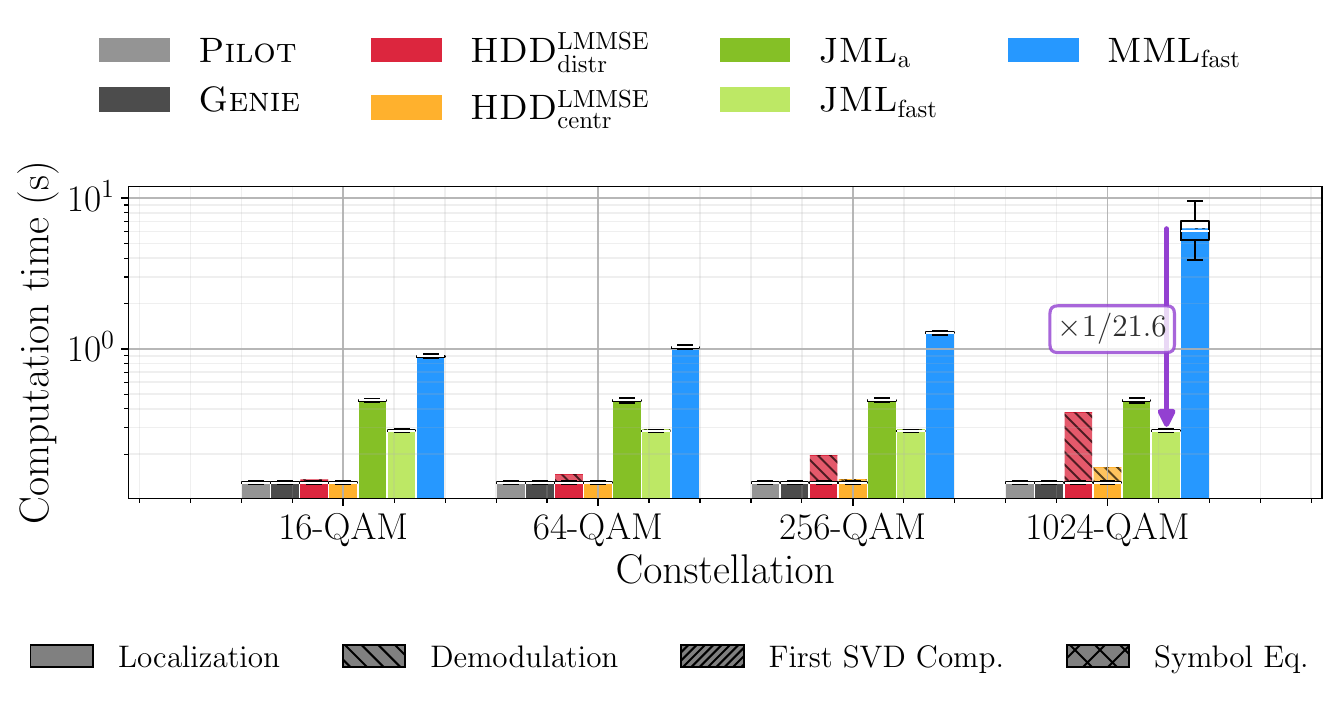}
    \caption{Mean computation time of the proposed methods and the considered baselines for different constellation sizes. Bars represent the mean localization time, with interquartile range (Q1--Q3) indicated by the box and the median shown as a white line. The hatched regions denote the cost of other processing steps, whose variability is not shown as it is negligible.}
    \label{fig:Time_vs_ConstSize}
\end{figure}

\subsubsection{Comparison with Baselines}

\autoref{fig:Time_vs_ConstSize} shows the mean computation time per position estimate for the proposed methods and the considered baselines, recorded during the simulations of \autoref{fig:RMSE_SNR_constellations}, with an additional $64$-\gls{qam} case included for completeness.
The following observations can be drawn:
\begin{itemize}
    \item The computational cost of the \gls{dd} baselines and $\MMLfast$ increases with constellation size, which is especially pronounced for the distributed \gls{dd} baseline and $\MMLfast$, as predicted by the asymptotic analysis.
    \item The cost of the proposed $\JMLa$ and $\JMLfast$ is completely independent of the constellation size, with computational requirements of similar order to those of the \gls{dd} baselines, while $\MMLfast$ is one order of magnitude slower. 
    
    For instance, for $1024$-\gls{qam}, $\MMLfast$ requires \SI{6.26}{\second}, $\JMLfast$ requires \SI{0.29}{\second} (a factor of \underline{$21.6$} faster than $\MMLfast$), $\HDDdistrmethod$ requires \SI{0.38}{\second}, and $\HDDcentrmethod$ requires \SI{0.16}{\second}, for the considered parameter settings.
\end{itemize}

These results validate the asymptotic analysis of Section~\ref{sec:Complexity_th}, and confirm that the proposed \gls{jml} estimators achieve substantially lower computational requirements than the \gls{mml} approach, while remaining comparable in complexity to the fastest \gls{dd} baselines.

\section{Conclusion}\label{sec:conclusion}

This paper investigates the joint exploitation of known pilot sequences and unknown data payloads for opportunistic source localization, in which a distributed \gls{srx} localizes the transmitting \gls{ue} by leveraging its communication signal as a signal of opportunity without interfering with the communication task. 
The proposed framework also extends to related configurations, including multistatic sensing and distributed transmitter scenarios.
A \gls{jml} framework is developed in which both the random channel coefficients and the unknown data symbols are treated as \glspl{np}. 
The optimal solution is derived in closed form and shown to be computationally intractable for typical \glspl{tfb}. 
Two tractable approximations are then proposed based on a theoretical analysis of the optimal solution, offering different performance-complexity tradeoffs.

\noindent The resulting estimators are \textit{constellation-agnostic}: regardless of the modulation scheme employed, both localization performance and computational requirements remain unchanged—a key advantage over \gls{dd} and marginal approaches.
Through \gls{mc} simulations, the proposed estimators demonstrate superior localization performance over existing baselines. Compared to pilot-only approaches, a \gls{rmse} reduction of up to a factor of $5.7$ is achieved. 
Compared to \gls{dd} baselines, an \gls{snr} gain of up to $\SI{4.8}{\decibel}$ is demonstrated at comparable computational complexity. 
Against the \gls{mml} approach of \cite{reniers_localization_2026}, similar or superior localization performance is achieved across most system parameter configurations, while reducing computational requirements by up to a factor of $21.6$ for large constellations.
Furthermore, unlike \gls{mml}, the proposed \gls{jml} estimators require no prior knowledge of the \gls{np} distributions.
Finally, a theoretical analysis of the proposed objective functions reveals that the \gls{jml} estimators operate as single-step hybrid \gls{toa}-\gls{tdoa} methods—combining range circle loci from the pilot term and hyperbolic loci from the data term.
Future work may consider the integration of Doppler estimation exploiting the time dimension, the impact of imperfect synchronization and multipath propagation on performance, the extension to multi-target multistatic sensing, and experimental validation in real-world environments.

\appendix

\subsection{Proof of \autoref{prop:JMLc1}}\label{app:JMLc_proof1}

\begin{proof}
Substituting $\channelcoeffEst(\UEposTest, \symbolMatTest)$ from \eqref{eq:channel_implicit_estimation} into $\channelcoeffTest$ in \eqref{eq:JMLc_LS}, and expanding the squares leads to
\begingroup\small
\begin{align}
     \UEposEst^{\JMLc} = & \argmin_{\UEposTest} \min_{\dataMatTest \in \C^{\numF\times\numD}}  \textcolor{gray}{\sum_{\RXindex,\Findex,\PDindex} \abs{\Obs[\RXindex,\Findex,\PDindex]}^2} \nonumber \\
    & - 2 \Re \bigg\{ \sum_{\RXindex} \sum_{\Findex,\PDindex}  \Obs^{*}[\RXindex,\Findex,\PDindex] \stMat(\UEposTest)[\RXindex,\Findex] \symbolMatTest[\Findex,\PDindex]  \nonumber \\
    & \frac{1}{\rgEnergy(\symbolMatTest)}\sum_{\Findex',\PDindex'} \Obs[\RXindex,\Findex',\PDindex']  \stMat^{*}(\UEposTest)[\RXindex,\Findex'] \symbolMatTest^{*}[\Findex', \PDindex']  \bigg\}   \nonumber \\
    & + \sum_{\RXindex} \textcolor{gray}{\sum_{\Findex,\PDindex} \bigg|\stMat(\UEposTest)[\RXindex,\Findex] \symbolMatTest[\Findex,\PDindex]\bigg|^{2}} \frac{1}{\rgEnergy^{\textcolor{gray}{2}}(\symbolMatTest)}  \nonumber \\
    & \bigg|\sum_{\Findex',\PDindex'} \stMat^{*}(\UEposTest)[\RXindex,\Findex'] \symbolMatTest^{*}[\Findex', \PDindex'] \Obs[\RXindex,\Findex',\PDindex']\bigg|^{2},
\end{align}
\endgroup
where the \textcolor{gray}{gray} terms are either independent of $\UEposTest$ and $\dataMatTest$, or cancel each other.
Simplifying yields
\begingroup\small
\begin{align}
     \UEposEst^{\JMLc} = & \argmin_{\UEposTest} \min_{\dataMatTest \in \C^{\numF\times\numD}}\nonumber \\
    & -  \frac{2}{\rgEnergy(\symbolMatTest)} \sum_{\RXindex} \Re \bigg\{ \bigg| \sum_{\Findex,\PDindex}  \Obs^{*}[\RXindex,\Findex,\PDindex] \stMat(\UEposTest)[\RXindex,\Findex] \symbolMatTest[\Findex,\PDindex] \bigg|^{2} \bigg\}  \nonumber \\
    & + \frac{1}{\rgEnergy(\symbolMatTest)}  \sum_{\RXindex} \bigg|\sum_{\Findex,\PDindex} \stMat^{*}(\UEposTest)[\RXindex,\Findex] \symbolMatTest^{*}[\Findex, \PDindex] \Obs[\RXindex,\Findex,\PDindex]\bigg|^{2}.
\end{align}
\endgroup
Equivalently maximizing the negated objective, and applying the identities $\abs{z} = \abs{z^*}$ and $\Re\{\abs{z}^2\} = \abs{z}^2$ for all $z \in \C$, completes the proof.
\end{proof}

\subsection{Proof of \autoref{prop:JMLc2}}\label{app:JMLc_proof2}

\begin{proof}
Expanding the pilot and data components in \eqref{eq:JMLc_LS_data_OF} by separating the sum over ($\PDindex$) into pilot ($\Pindex$) and data ($\Dindex$) contributions, and substituting the quantities introduced in \autoref{tab:notations_JMLc} yields
\begin{align}
     \OF(\UEposTest, \dataVecTest)  & = \frac{\sum_{\RXindex} \abs{\channelcoeffEstPilotsUn{\RXindex}(\UEposTest) + \dataVecEstcc{\RXindex}^{\Herm}(\UEposTest) \dataVecTest }^2 }{\pilotEnergy + \dataVecTest^{\Herm}\dataVecTest} \\
    &  = \frac{\norm{\channelcoeffEstPilotsU(\UEposTest) + \dataMatEstccConcat^{\Herm}(\UEposTest)\dataVecTest}^2}{\pilotEnergy + \dataVecTest^{\Herm} \dataVecTest}.
\end{align}
Developing the squared norm reveals the aggregate Gram matrix $\dataMatEstccCov(\UEposTest)$ and the unnormalized soft data estimates with channel precompensated by pilot observations (see \autoref{tab:notations_JMLc}), $\dataVecEstccP(\UEposTest)$:
\begin{align}
    & \OF(\UEposTest, \dataVecTest) = (\pilotEnergy + \dataVecTest^{\Herm}\dataVecTest)^{-1}  \times \Big(\channelEnergy(\UEposTest)  \nonumber \\
    & \hspace{0.5cm} + 2 \Re \{ \channelcoeffEstPilotsU^{\Herm} \dataMatEstccConcat^{\Herm}(\UEposTest)\dataVecTest \}  + \dataVecTest^{\Herm} \dataMatEstccConcat(\UEposTest) \dataMatEstccConcat^{\Herm}(\UEposTest) \dataVecTest \Big) \\
    & = \frac{\channelEnergy(\UEposTest) + \dataVecEstccP^{\Herm}(\UEposTest) \dataVecTest + \dataVecTest^{\Herm} \dataVecEstccP(\UEposTest)  + \dataVecTest^{\Herm} \dataMatEstccCov(\UEposTest) \dataVecTest }{\pilotEnergy + \dataVecTest^{\Herm}\dataVecTest}. \label{eq:OF_dev1}
\end{align}
Omitting the dependency on $\UEposTest$ for brevity, let us define the numerator of \eqref{eq:OF_dev1} as $\OFnum(\dataVecTest)$ and the denominator as $\OFdenom(\dataVecTest)$, such that $\OF(\dataVecTest) \triangleq \OFnum(\dataVecTest) / \OFdenom(\dataVecTest)$.
The stationarity condition is
\begin{equation}
    \frac{\partial \OF(\dataVecTest)}{\partial\dataVecTest^{*}} = \frac{\frac{\partial \OFnum(\dataVecTest)}{\partial\dataVecTest^{*}}\OFdenom(\dataVecTest) -\OFnum(\dataVecTest) \frac{\partial \OFdenom(\dataVecTest)}{\partial\dataVecTest^{*}}}{\OFdenom^2(\dataVecTest)} = 0.
\end{equation}
Note that the other Wirtinger derivative $\frac{\partial \OF(\dataVecTest)}{\partial\dataVecTest}=0$ leads to the same condition by conjugate symmetry.
The individual derivatives are given by
\begin{equation}
    \frac{\partial \OFnum(\dataVecTest)}{\partial\dataVecTest^{*}} = \dataMatEstccCov \dataVecTest + \dataVecEstccP \quad \text{and} \quad 
    \frac{\partial \OFdenom(\dataVecTest)}{\partial\dataVecTest^{*}} = \dataVecTest.
\end{equation}
Finally, substituting into the stationarity condition yields
\begin{equation}
    \left( \dataMatEstccCov \dataVecTest + \dataVecEstccP \right)\OFdenom(\dataVecTest) = \OFnum(\dataVecTest) \dataVecTest,
\end{equation}
and dividing both sides by $\OFdenom(\dataVecTest)$, which is strictly positive, completes the proof.
\end{proof}

\subsection{Proof of \autoref{prop:JMLa}}\label{app:JMLa_proof}

\begin{proof}
The pilot term follows directly from \eqref{eq:JMLc_LS_data} by restricting to the pilot sequence, omitting the maximization over $\dataMatTest$, and substituting \eqref{eq:pilotObsEq}.
For the data term, the derivation proceeds similarly as in \hyperref[app:JMLc_proof1]{Appendix~\ref*{app:JMLc_proof1}}. 
Starting from the \gls{ls} formulation \eqref{eq:JMLc_LS} and substituting $\channelcoeffEstPilots$ for $\channelcoeffTest$, minimizing over $\dataMatTest$:
\begingroup \small
\begin{align}
    \sum_{\RXindex=0}^{\numRX-1} \sum_{\Findex=0}^{\numF-1} \sum_{\Dindex=0}^{\numD-1} \abs{\dataObs[\RXindex,\Findex,\Dindex] - \channelcoeffEstPilots[\RXindex] \stMat(\UEposTest)[\RXindex,\Findex] \dataMatTest[\Findex,\Dindex] }^2
\end{align}
\endgroup
is separable across $\Findex$ and $\Dindex$. It then yields the following implicit estimates via Wirtinger calculus:
\begin{equation}
    \dataMatEst[\Findex,\Dindex](\UEposTest) = \frac{\sum_{\RXindex=0}^{\numRX-1} \channelcoeffEstPilots^{*}[\RXindex] \stMat^{*}(\UEposTest)[\RXindex,\Findex] \dataObs[\RXindex,\Findex,\Dindex] }{\sum_{\RXindex=0}^{\numRX-1} \abs{\channelcoeffEstPilots[\RXindex] \stMat(\UEposTest)[\RXindex,\Findex]}^2},
\end{equation}
where the denominator simplifies to $\channelConstructEnergy(\UEposTest) \triangleq \norm{\channelcoeffEstPilots}^2$, since $\abs{\stMat(\UEposTest)[\RXindex,\Findex]} = 1 \, \forall \UEposTest$.
The remainder of the derivation follows by analogy with \hyperref[app:JMLc_proof1]{Appendix~\ref*{app:JMLc_proof1}}, with the following correspondence: the role of the channel $\channelcoeffTest$ and data $\dataMatTest$ are interchanged, and the substitutions summarized in \autoref{tab:proof_substitutions} apply.
\begin{table}[ht]
    \centering
    \caption{Variable correspondence with \hyperref[app:JMLc_proof1]{Appendix~\ref*{app:JMLc_proof1}}.}
    \label{tab:proof_substitutions}
    \begin{tabular}{|cc|}
        \hline
        \hyperref[app:JMLc_proof1]{Appendix~\ref*{app:JMLc_proof1}} & Present proof \\
        \hline
        $\rgEnergy(\symbolMatTest)$ & $\channelConstructEnergy(\UEposTest)$ \\
        $\symbolMatTest$ & $\channelcoeffEstPilots(\UEposTest)$ \\
        $\Obs$ & $\dataObs$ \\
        Summation over $\Findex, \PDindex$ ($\Findex', \PDindex'$) & Summation over $\RXindex$ ($\RXindex'$) \\
        Summation over $\RXindex$ & Summation over $\Findex, \Dindex$ \\
        \hline
    \end{tabular}
\end{table}
\end{proof}

\subsection{Proof of \autoref{prop:JMLfast}}\label{app:JMLfast_proof}
\begin{proof}
For notational convenience, let us introduce $\dataObsMat{\Findex} \triangleq \mbox{$\dataObs[:,\Findex,:]$} \in \C^{\numRX\times\numD}$ and $\channelConstructVec{\Findex}(\UEposTest) \triangleq \channelConstruct[:,\Findex](\UEposTest) \in \C^{\numRX\times1}$.
The data component of the objective function for subcarrier $\Findex$ then becomes
\begin{align}
    \ell_{\mathrm{D},\Findex}(\UEposTest) & \triangleq  \sum_{\Dindex=0}^{\numD-1} \Bigg| \sum_{\RXindex=0}^{\numRX-1} \dataObs^{*}[\RXindex,\Findex,\Dindex] \channelConstruct(\UEposTest)[\RXindex,\Findex] \Bigg|^2 \\
    & = \norm{\dataObsMat{\Findex}^{\Herm} \channelConstructVec{\Findex}(\UEposTest)}^2  = \left(\dataObsMat{\Findex}^{\Herm} \channelConstructVec{\Findex}(\UEposTest)\right)^{\Herm}  \left(\dataObsMat{\Findex}^{\Herm} \channelConstructVec{\Findex}(\UEposTest)\right) \\
    & = \channelConstructVec{\Findex}^{\Herm}(\UEposTest) \dataObsMat{\Findex} \dataObsMat{\Findex}^{\Herm} \channelConstructVec{\Findex}(\UEposTest) = \channelConstructVec{\Findex}^{\Herm}(\UEposTest) \dataObsMatCov{\Findex} \channelConstructVec{\Findex}(\UEposTest),
\end{align}
where $\dataObsMatCov{\Findex} \triangleq \dataObsMat{\Findex} \dataObsMat{\Findex}^{\Herm} \in \C^{\numRX \times \numRX}$ is the sample correlation matrix of the observations associated with subcarrier $\Findex$, which compresses the $\numD$ dimension.

We now define a matrix $\dataObsMatCov{\Findex}^{0.5} \in \C^{\numRX \times \numA}$ such that $\dataObsMatCov{\Findex} = \dataObsMatCov{\Findex}^{0.5} (\dataObsMatCov{\Findex}^{0.5})^{\Herm}$, where $\numA$ is an arbitrary dimension.
We then have
\begin{align}
     \ell_{\mathrm{D},\Findex}(\UEposTest) & = \channelConstructVec{\Findex}^{\Herm}(\UEposTest) \dataObsMatCov{\Findex}^{0.5} (\dataObsMatCov{\Findex}^{0.5})^{\Herm} \channelConstructVec{\Findex}(\UEposTest) \\
     & = \norm{(\dataObsMatCov{\Findex}^{0.5})^{\Herm} \channelConstructVec{\Findex}(\UEposTest)}^{2}.
\end{align}
Note that $\dataObsMatCov{\Findex}^{0.5} = \dataObsMat{\Findex}$ is a trivial choice.
Here we present a more efficient one.
Let $\dataObsMat{\Findex} = \SVDMatLeft{\Findex} \SVDMatVals{\Findex} \SVDMatRight{\Findex}^{\Herm}$ be the \gls{svd} of $\dataObsMat{\Findex}$, where $\SVDMatLeft{\Findex} \in \C^{\numRX\times\numRX}$ (resp. $\SVDMatRight{\Findex} \in \C^{\numD\times\numD}$) is a unitary matrix whose columns are the left (resp. right) singular vectors of $\dataObsMat{\Findex}$, and $\SVDMatVals{\Findex} \in \C^{\numRX\times\numD}$ is a rectangular matrix whose diagonal entries are the singular values of $\dataObsMat{\Findex}$ and all off-diagonal entries are zero.
We now show that $\dataObsMatCov{\Findex}^{0.5} = \SVDMatLeft{\Findex} \SVDMatVals{\Findex}$ is a valid choice.
Substituting the \gls{svd} of $\dataObsMat{\Findex}$ into the definition of the correlation matrix yields
\begin{align}
    \dataObsMatCov{\Findex} & = \dataObsMat{\Findex} \dataObsMat{\Findex}^{\Herm} = \left(  \SVDMatLeft{\Findex} \SVDMatVals{\Findex} \SVDMatRight{\Findex}^{\Herm} \right) \left(  \SVDMatLeft{\Findex} \SVDMatVals{\Findex} \SVDMatRight{\Findex}^{\Herm} \right)^{\Herm} \\
    &  = \SVDMatLeft{\Findex} \SVDMatVals{\Findex} \SVDMatRight{\Findex}^{\Herm}  \SVDMatRight{\Findex} \SVDMatVals{\Findex}^{\Herm} \SVDMatLeft{\Findex}^{\Herm} = \SVDMatLeft{\Findex} \SVDMatVals{\Findex}  \SVDMatVals{\Findex}^{\Herm} \SVDMatLeft{\Findex}^{\Herm} \label{eq:JMLfast_svd_evd} \\
    & = \left( \SVDMatLeft{\Findex} \SVDMatVals{\Findex} \right) \left( \SVDMatLeft{\Findex} \SVDMatVals{\Findex} \right)^{\Herm}, \label{eq:JMLfast_unitary}
\end{align}
where \eqref{eq:JMLfast_unitary} exploits the unitary property of $\SVDMatRight{\Findex}$ (i.e., $\SVDMatRight{\Findex}^{\Herm}\SVDMatRight{\Findex} = \I{\numD}$).
Note that \eqref{eq:JMLfast_svd_evd} also reveals the well-known relationship between the \gls{svd} of the signal $\dataObsMat{\Findex}$ and the \gls{evd} of its correlation matrix:
$\dataObsMatCov{\Findex} = \boldsymbol{V}_{\Findex} \boldsymbol{\Lambda}_{\Findex} \boldsymbol{V}_{\Findex}^{\Herm}$, with $\boldsymbol{V}_{\Findex} = \SVDMatLeft{\Findex}$ and $\boldsymbol{\Lambda}_{\Findex} = \SVDMatVals{\Findex}  \SVDMatVals{\Findex}^{\Herm}$.
By identification, $\dataObsMatCov{\Findex}^{0.5} = \SVDMatLeft{\Findex} \SVDMatVals{\Findex}$ is a valid choice, yielding
\begin{align}
     \ell_{\mathrm{D},\Findex}(\UEposTest) &= \norm{\SVDMatVals{\Findex}^{\Herm} \SVDMatLeft{\Findex}^{\Herm} \channelConstructVec{\Findex}(\UEposTest)}^{2} \\
     & = \sum_{\SVindex=0}^{\numSV-1} \abs{\SVDValVec{\Findex}[\SVindex]}^2 \abs{\sum_{\RXindex=0}^{\numRX-1} \SVDMatLeft{\Findex}^{\Herm}[\SVindex,\RXindex] \channelConstructVec{\Findex}[\RXindex]}^2,
\end{align}
where $\SVDValVec{\Findex} \triangleq \begin{bmatrix} \SVDMatVals{\Findex}[0,0] & \cdots & \SVDMatVals{\Findex}[\numSV-1,\numSV-1] \end{bmatrix}^{\Trans} \in \R^{\numSV\times1}_{+}$ and $\numSV = \min(\numRX,\numD)$ denotes the number of singular values.
\end{proof}

{
\begingroup
\renewcommand{\baselinestretch}{0.98}\selectfont
\bibliographystyle{IEEEtran}
\bibliography{nourl,references}
\endgroup
}

\end{document}